\documentclass[10pt,journal,doublecolumn]{IEEEtran}

\usepackage{threeparttable,placeins,makecell,stfloats,cite}
\usepackage{threeparttable,arydshln}
\usepackage{amsmath,amssymb,bm}
\usepackage{mathtools}
\usepackage{graphicx}
\usepackage[usenames]{color}
\usepackage{amsfonts}
\usepackage{latexsym}
\usepackage{multirow}
\usepackage{caption}
\usepackage{rotating}
\usepackage{float}
\usepackage{url}
\usepackage{cite}
\usepackage[noend]{algpseudocode}
\usepackage{algorithmicx,algorithm}
\usepackage[noend]{algpseudocode}
\usepackage{diagbox}
\usepackage{cuted}
\algrenewcommand\algorithmicthen{}
\algrenewcommand\algorithmicdo{}

\newtheorem{lemma}{\textit{Lemma}}

\newtheorem{proposition}{\textit{Proposition}}

\usepackage{graphicx} 

\usepackage[hidelinks]{hyperref}
\usepackage[noabbrev,nameinlink]{cleveref}

\usepackage{overpic} 
\usepackage{subfigure}

\IEEEoverridecommandlockouts

\begin{document}

\title{Enabling Flexible AFDM-ISAC Design:\\ When Ambiguity
Shaping Meets PAPR Control}
\author{Lingsheng~Meng, Yong Liang~Guan, \textit{Senior Member, IEEE}, Zilong~Liu, \textit{Senior Member, IEEE}, \\Yirui~Luo, \textit{Graduate Student Member, IEEE}, and Pingzhi Fan, \textit{Life Fellow, IEEE}\vspace{-1cm}

\thanks{Lingsheng Meng, Yong Liang Guan, and Yirui~Luo are with the School of Electrical and Electronics Engineering, Nanyang Technological University, Singapore 639798 (e-mail: meng0071@e.ntu.edu.sg; eylguan@ntu.edu.sg; yirui001@e.ntu.edu.sg). Zilong Liu is with the School of Computer Science and Electronics Engineering, University of Essex, Colchester CO4 3SQ, UK (e-mail: zilong.liu@essex.ac.uk). Pingzhi Fan is with the Key Laboratory of Information Coding and Transmission of Sichuan Province, CSNMT International Cooperation Research Centre (MoST), Southwest Jiaotong University, Chengdu 610031, China (e-mail: p.fan@ieee.org).}
}
\maketitle

\begin{abstract}
Affine frequency division multiplexing (AFDM) has emerged as an enabling waveform for integrated sensing and communication (ISAC) due to its intrinsic chirp signaling nature.
Nevertheless, the practicality of AFDM-ISAC systems needs to address two major technical challenges, i.e., high ambiguity function (AF) sidelobes and high peak-to-average power ratio (PAPR).
\textcolor{black}{By exploiting the reserved chirp-subcarrier (RCS) symbols and per-subcarrier pre-chirp parameters, we develop a flexible and unified AFDM waveform design framework for AF shaping and PAPR control.}
The proposed framework supports three tunable modes: AF shaping via weighted integrated sidelobe level (ISL) minimization, PAPR minimization, and joint AF shaping and PAPR control under a prescribed PAPR constraint.
To solve the formulated nonconvex problem and to accommodate the discrete-phase constraints on the pre-chirp parameters, a joint ISL-PAPR discrete-phase majorization-minimization (JIPD-MM) algorithm is developed.
\textcolor{black}{Simulation results verify the effectiveness of the proposed framework under all the three design modes, with the benchmark comparisons conducted at similar effective spectral efficiencies.
It is shown that the resulting weighted-ISL and PAPR gains lead to improved weak-target detectability in multi-target scenarios and lower bit error rate (BER) under power-amplifier (PA) nonlinearity.}
\end{abstract}


\section{Introduction}
Integrated sensing and communication (ISAC), as a major use case for the upcoming sixth-generation (6G) systems \cite{Liu2020JRCDesign,GonzalezPrelcic2024ISACRevolution,Wei2023ISACSignalsSurvey}, has attracted tremendous attention from academia, industry, the International Telecommunication Union (ITU) and the 3rd Generation Partnership Project (3GPP)~ \cite{ITUR2023IMT2030FrameworkDraft, 3GPPref}. By allowing both the sensing and communication modules to share the same hardware and spectrum resources as well as a majority of signal processing tasks, ISAC not only supports flexible performance trade-offs and efficient resource utilization but can also mutually enhance the two functionalities \cite{Liu2022ISACJSAC}.

To support emerging ISAC applications in a diverse set of use cases and quality-of-service requirements, such as vehicle-to-everything (V2X) and low-altitude economy, a central question is \textit{how to design highly flexible and adaptable waveforms that are resilient to doubly selective channels while permitting accurate sensing}. The legacy orthogonal frequency division multiplexing (OFDM) has been extensively studied for ISAC, but it is sensitive to Doppler-induced inter-carrier interference (ICI) in high-mobility environments, thus suffering from degraded performance for both communication and sensing \cite{Wang2006DopplerOFDM}.

Recently, affine frequency division multiplexing (AFDM) has emerged as a promising waveform for high-mobility communications and ISAC \cite{AFDMOR,Luo2024AFDMSCMA,Sui2025MultiFunctionalChirp,Luo2025AFDMISACADD,Tek2025SecureAFDM}. In principle, AFDM modulates its data symbols through a set of orthogonal chirp subcarriers that are induced by the discrete affine Fourier transform (DAFT). By properly tuning the chirp rate specified by the post-chirp parameter $c_1$ according to the maximum Doppler frequency, AFDM is able to attain full channel diversity and hence significantly enhanced error rate performance. Moreover, AFDM subsumes both OFDM and orthogonal chirp division multiplexing (OCDM) \cite{Ouyang2016OCDM} and can be efficiently generated through minimum modification of an OFDM transmitter. As such, AFDM is backward-compatible and allows seamless integration into the existing wireless infrastructure that is dominated by OFDM \cite{AFDMOR,Rou2026AFDMEvolvingOFDM}. 

From the ISAC perspective, AFDM has attracted increasing attention. Since sensing is often realized through matched filtering (MF) in the time or DAFT domain, the ambiguity function (AF) of the transmitted waveform plays a central role~\cite{Ye2022,Liu2025SensingCommSignals}. Existing AFDM-ISAC designs mainly focus on pilot-aided sensing structures. For example, a single dedicated pilot framework was proposed in~\cite{AFDM_isac}, where pilot--data separability is ensured by allocating surrounding null-guard subcarriers, while superimposed-pilot strategies were considered in~\cite{Zheng2025AFDMSuperimposedPilots,Zhang2025AFDMEnabledISACJSAC} to improve the spectral efficiency. Another research line is to characterize the AF properties of AFDM signals with typical pilot structures~\cite{YHR_AFDM_AF_Analysis}. The AF behavior of data-bearing AFDM signals under different chirp-rate configurations has also been investigated~\cite{Bedeer2025AFDMAmbiguity}. Meanwhile, some studies also explored sensing using the data-bearing AFDM signal itself~\cite{Ni2025AFDMISAC,AFDM_isac}. In this case, however, the sensing performance may be compromised due to the random data. 

Beyond the ISAC perspective, AFDM in practical systems also faces a critical transmitter-side challenge, namely high peak-to-average power ratio (PAPR)~\cite{Bazzi2023TunablePAPR,Wang2026DRIP}. Similar to OFDM, this issue arises from the multicarrier nature of AFDM, and a high PAPR could drive the power amplifier (PA) to operate in the nonlinear region, resulting in nonlinear distortion~\cite{LopezPerez2022RANEE}. For OFDM, a number of PAPR reduction techniques have been developed, including selective mapping, precoding, and reserved-tone based schemes \cite{Gopi2020OptimizedSLM,Feng2025,OFDMPAPR}. In contrast, PAPR reduction for AFDM remains much less explored. A DAFT-spread-based affine frequency division multiple access scheme was proposed in~\cite{11185309}, but it is primarily an uplink multiple-access architecture rather than a waveform design framework for AFDM systems. The grouped pre-chirp selection (GPS) method in~\cite{AFDM_PAPR} reduces the PAPR of AFDM systems by selecting different DAFT-domain pre-chirp parameters $c_2$ across subcarriers, but the resulting pre-chirp parameters need to be transmitted to the receiver.

In summary, existing AFDM-ISAC studies mainly focus on pilot-aided structures and related analyses, with little consideration on data-bearing waveform design methods. The current PAPR reduction methods for AFDM either rely on a prescribed spreading structure and thus cannot flexibly optimize an arbitrarily given AFDM waveform instance, or depend on pre-chirp parameter selection at the cost of additional side-information overhead. \textcolor{black}{An effective and tunable design approach that can flexibly carry out both AF shaping and PAPR control according to different ISAC and transmitter requirements is still missing.}

\textcolor{black}{To fill the aforementioned research gaps, the major novelty of this work lies in a flexible and unified AFDM-ISAC waveform design framework that can be readily tuned according to practical communication and/or sensing needs. Unlike conventional reserved-tone designs that mainly target PAPR reduction~\cite{OFDMPAPR}, the proposed reserved chirp-subcarrier (RCS) approach conceives a more challenging research problem for both AF-sidelobe shaping and PAPR control. In contrast to OFDM, AFDM also allows the pre-chirp parameter to vary across subcarriers, i.e., using per-subcarrier $c_{2,m}$ rather than a single $c_2$. In the proposed framework, the RCS symbols and the per-subcarrier pre-chirp parameters are incorporated into a unified DAFT-domain formulation, enabling flexible and reconfigurable AFDM waveform design according to different system requirements.}

Specifically, the main contributions of this work are summarized as follows.
\begin{itemize}

  \item \textcolor{black}{We present a unified and reconfigurable AFDM waveform optimization framework for AF shaping over a prescribed low-ambiguity zone (LAZ) and PAPR control. Its design flexibility is reflected in both the waveform-design configuration and the performance objective. The framework can optimize the RCS symbols alone or further incorporate the per-subcarrier pre-chirp parameters according to the desired performance and signaling overhead. It also supports three operating modes, including an AF shaping mode for integrated sidelobe level (ISL) suppression, a PAPR minimization mode for oversampled AFDM waveforms, and a joint AF shaping and PAPR control mode that minimizes ISL under a prescribed PAPR constraint, thereby allowing a flexible balance between sensing and communication.}

  \item We develop a joint ISL-PAPR discrete-phase majorization-minimization (JIPD-MM) algorithm to tackle the formulated AFDM waveform optimization problems. The challenge arises from the nonconvex weighted-ISL objective, the non-smooth maximum in the PAPR metric, and the finite-alphabet constraints on the pre-chirp parameters. To handle this, smooth majorization-minimization (MM) surrogates are constructed for the weighted-ISL objective and the PAPR metric. The PAPR constraint is incorporated through a penalty formulation, while the discrete-phase constraints are handled through an initialization stage followed by the main optimization stage.

 \item Simulation results verify the effectiveness of the proposed framework for both sensing and communication. \textcolor{black}{The comparisons explicitly account for the spectral-efficiency loss caused by RCS allocation and, when applicable, the signaling overhead of the optimized per-subcarrier pre-chirp parameters through the effective spectral efficiency. The proposed JIPD-MM algorithm achieves substantial ISL/PAPR reduction in the AF shaping and PAPR minimization modes, respectively. In the joint AF shaping and PAPR control mode, the proposed design can effectively suppress AF sidelobes while maintaining the waveform PAPR below a prescribed target. These waveform-level improvements are also reflected in the system-level results, including improved weak-target detectability and bit error rate (BER) advantage under PA nonlinearity.}
\end{itemize}

\textcolor{black}{\textit{Notations:} The operators $(\cdot)^{\mathsf{T}}$, $(\cdot)^{*}$, and $(\cdot)^{\mathsf{H}}$ denote transpose, complex conjugate, and conjugate transpose, respectively. We use $|\cdot|$ for the modulus of a complex scalar and $\|\cdot\|$ for the Euclidean norm. The sets of complex numbers and integers are denoted by $\mathbb{C}$ and $\mathbb{Z}$, respectively. We use $\mathbf{I}_N$ for the $N\times N$ identity matrix and $\mathbb{E}\{\cdot\}$ for expectation. The operators $\Re\{\cdot\}$ and $\Im\{\cdot\}$ extract the real and imaginary parts, and the imaginary unit is denoted by $j=\sqrt{-1}$. The notation $\operatorname{diag}(\bm\rho)$ denotes the diagonal matrix whose principal diagonal is $\bm\rho$, and $\operatorname{vec}(\cdot)$ denotes the stacking vectorization of a matrix.}

\vspace{-0.2cm}

\section{Preliminaries}
\label{Preliminaries}

In this work, we consider a typical downlink monostatic ISAC system in which a base station (BS) serves user equipment (UE) and simultaneously senses surrounding targets using echoes of the AFDM waveform, as illustrated in the right-hand side of Fig.~\ref{Fig2}. The ISAC-BS is assumed to operate in full duplex, with sufficient transmit--receive isolation and self-interference cancellation such that the residual self-interference can be neglected~\cite{Liu2025SensingCommSignals,OFDMPAPR,Zhang2025AFDMEnabledISACJSAC}. \textcolor{black}{For sensing, the received signal is modeled as a superposition of delayed and Doppler-shifted echoes of the transmitted AFDM waveform, weighted by the target complex reflection coefficients and corrupted by additive white Gaussian noise (AWGN). For communication, both the AWGN channel and the doubly selective channel are considered, where the latter consists of multiple propagation paths characterized by complex path gains, delays, and Doppler shifts.}

Let $N$ denote the number of AFDM subcarriers (chirps) within one AFDM symbol of duration $T$. The DAFT-domain transmit vector is given by $\mathbf{x} = [x[0],x[1],\ldots,x[N-1]]^{\mathsf{T}} \in \mathbb{C}^{N}$, where $x[m]$ is the complex symbol placed on the $m$-th AFDM subcarrier. In conventional AFDM, the DAFT parameters $(c_1,c_2)$ are common to all subcarriers. In this work, we keep a common post-chirp parameter $c_1$ (also referred to as the chirping rate), but allow the per-subcarrier pre-chirp parameters (also referred to as the per-subcarrier initial-phase parameters), to vary with the subcarrier index $m$~\cite{AFDM_PAPR}, which we denote by $\{c_{2,m}\}_{m=0}^{N-1}$. Standard AFDM corresponds to the special case $c_{2,m} \equiv c_2$. Then, the transmit signal $\mathbf{s} = [s[0], s[1], \ldots, s[N-1]]^{\mathsf{T}}$ is generated via the inverse DAFT~\cite{AFDMOR}
\begin{equation}
 s[n]
 = \frac{1}{\sqrt{N}} \sum_{m=0}^{N-1}
 x[m]\,
 \exp\!\left( j 2\pi ( c_1 n^2 + \tfrac{m}{N} n + c_{2,m} m^2 ) \right).
 \label{eq:AFDM_IDAFT_disc}
\end{equation}
\textcolor{black}{The expression in \eqref{eq:AFDM_IDAFT_disc} can be rewritten in matrix form as}
\begin{equation}
 \mathbf{s} = \mathbf{A}\mathbf{x},
 \label{eq:AFDM_IDAFT_matrix}
\end{equation}
where $\mathbf{A}
 =
 \boldsymbol{\Lambda}_{c_1}\,\mathbf{F}_{N}^{\mathsf{H}}\,\boldsymbol{\Lambda}_{c_2}$, $[\mathbf{F}_{N}]_{n,m} = \frac{1}{\sqrt{N}} e^{-j2\pi nm/N}$,
$\boldsymbol{\Lambda}_{c_1}=\operatorname{diag}\!\big(e^{j2\pi c_1 0^2},e^{j2\pi c_1 1^2},\ldots,e^{j2\pi c_1 (N-1)^2}\big)$,
and $\boldsymbol{\Lambda}_{c_2}=\operatorname{diag}\!\big(e^{j2\pi c_{2,0} 0^2},e^{j2\pi c_{2,1} 1^2},\ldots,e^{j2\pi c_{2,N-1} (N-1)^2}\big)$.

In AFDM system, each AFDM symbol is appended with a chirp-periodic prefix (CPP)~\cite{AFDM_isac}. The CPP reduces to the conventional cyclic prefix (CP) like that of an OFDM symbol when $2 N c_1$ is an integer and $N$ is even~\cite{AFDMOR}. For simplicity, we restrict our attention to this CP-equivalent case throughout the paper. The CPP/CP duration is chosen to be larger than both the maximum delay of the communication channel and the maximum round-trip delay of the sensing targets.

Then, the AF of the AFDM transmit waveform at a delay--Doppler pair $(\tau,\mu)$ is defined as \cite{Ye2022,Zhang2025AFDMEnabledISACJSAC}
\begin{equation}
 A_{\tau,\mu}^{{\mathbf{s}}}
 =
 \mathbf{s}^{\mathsf{H}}
 \mathbf{U}_{\tau,\mu}
 \mathbf{s},
 \label{eq:AF_def_time_domain}
\end{equation}
where $\mathbf{U}_{\tau,\mu}
 =
 \mathbf{J}_{\tau}\mathbf{D}(\mu)$.
The delay-shift matrix $\mathbf{J}_{\tau}$ is given by
\begin{align*}
 &[\mathbf{J}_{\tau}]_{i,j}
 =
 \begin{cases}
 1, & i-j \equiv \tau \; (\mathrm{mod}\; N),\\[1ex]
 0, & \text{otherwise},
 \end{cases}
 \\ &\text{for } i,j = 0,\ldots,N-1,
\end{align*}
and the Doppler-shift matrix is
\begin{equation*}
 \mathbf{D}(\mu)
 =
 \operatorname{diag}\left(
 e^{-j 2\pi \mu \cdot 0/N},
 e^{-j 2\pi \mu \cdot 1/N},
 \ldots,
 e^{-j 2\pi \mu \cdot (N-1)/N}
 \right).
\end{equation*}
Here, $\tau \in \mathbb{Z}$ is the delay index in samples, the round-trip delay is $\tau \Delta t$ with $\Delta t = T/N$, corresponding to a point target at range $R \approx c \tau \Delta t/2$, where $c$ is the speed of light. The normalized Doppler index $\mu \in\mathbb{R}$ is related to the physical Doppler shift $f_{\mathrm{D}}$ (in Hz) by $\mu=f_{\mathrm{D}}T=Nf_{\mathrm{D}}\Delta t$. For a carrier frequency $f_{\mathrm{c}}$, this Doppler shift corresponds to a radial velocity $v$ through $f_{\mathrm{D}} \approx 2 v f_{\mathrm{c}}/c$ in a monostatic sensing link.

In practical scenarios, the maximum delay and maximum Doppler shift may be much smaller than the sequence duration and signal bandwidth \cite{1214833, meng2024}. Hence, it is of strong practical significance to shape the AF within a prescribed delay--Doppler LAZ \cite{Ye2022,meng2025}. Let the delay index satisfy $-\tau_{\max} \le \tau \le \tau_{\max}$ and the normalized Doppler frequency lie in a closed interval $[\mu_{\min},\mu_{\max}]$. To approximate an integral over Doppler, we partition $[\mu_{\min},\mu_{\max}]$ into $L_{\mu}$ grid points
\[
 \mu_{q}
 =
 \mu_{\min} + (q-1)\,\Delta\mu,
 \qquad
 q = 1,\ldots,L_{\mu},
\]
where $\Delta\mu = (\mu_{\max}-\mu_{\min})/(L_\mu-1)$. We then collect all delay--Doppler sampling points of interest in the finite set
\[
 \mathcal{A}
 =
 \left\{
 (\tau,\mu_{q})
 \,\big|\,
 -\tau_{\max} \le \tau \le \tau_{\max},
 \; 1 \le q \le L_{\mu}
 \right\}
 \setminus \{(0,0)\}.
\]
Given nonnegative weights $\{ w_{\tau,\mu_{q}} \}_{(\tau,\mu_{q})\in\mathcal{A}}$, the weighted ISL of the AF is defined as
\begin{equation}
 J_{\mathrm{ISL}}
 =
 \sum_{(\tau,\mu_{q})\in\mathcal{A}}
 w_{\tau,\mu_{q}}\,
 \left | A_{\tau,\mu_{q}}^{{\mathbf{s}}} \right |^{2}.
 \label{eq:ISL_def}
\end{equation}
\textcolor{black}{The AF characterizes the response of the sensing receiver to delayed and Doppler-shifted replicas of the transmitted waveform. High AF sidelobes associated with a strong target may mask weaker targets and thus degrade detection performance in multitarget scenarios~\cite{Ye2022,Liu2025SensingCommSignals}. We adopt the weighted ISL in \eqref{eq:ISL_def} as the waveform-design metric because it directly quantifies the aggregate AF-sidelobe energy over the LAZ~\cite{Ye2022,meng2025}. In the numerical results, we further evaluate the sensing benefit of the sidelobe reduction in terms of the detection probability $P_{\mathrm{D}}$ and false-alarm probability $P_{\mathrm{FA}}$. Here, $P_{\mathrm{D}}$ denotes the probability that a target is correctly detected at its corresponding delay--Doppler bin, while $P_{\mathrm{FA}}$ denotes the probability of declaring a target at a given delay--Doppler bin when no valid target echo is received~\cite{Skolnik2008Radar}.}

The PAPR of the AFDM transmit waveform $\mathbf{s}$ is defined as
\begin{equation}
\mathrm{PAPR}(\mathbf{s})
=
\frac{\displaystyle
\max_{0 \le n < N L_{\mathrm{P}}}
\left|\tilde{s}^{(\mathrm{P})}[n]\right|^{2}}
{\displaystyle
\frac{1}{N L_{\mathrm{P}}}
\sum_{n=0}^{N L_{\mathrm{P}}-1}
\left|\tilde{s}^{(\mathrm{P})}[n]\right|^{2}},
\label{eq:PAPR_def}
\end{equation}
where $\tilde{\mathbf{s}}^{(\mathrm{P})}\in \mathbb{C}^{N L_{\mathrm{P}}}$ denotes an $L_{\mathrm{P}}$-times oversampled discrete-time representation of the Nyquist-rate waveform $\mathbf{s}$. For PAPR evaluation, it is obtained through DFT-based interpolation, implemented by zero padding the DFT of $\mathbf{s}$ and applying an $N L_{\mathrm{P}}$-point IDFT \cite{ORPAPR,OFDMPAPR}. Specifically,
\begin{equation}
\tilde{\mathbf{s}}^{(\mathrm{P})}
=
\mathbf{P}_{L_{\mathrm{P}}}\mathbf{s},
\label{eq:s_tilde_matrix}
\end{equation}
where $\mathbf{P}_{L_{\mathrm{P}}}
=
\sqrt{L_{\mathrm{P}}}\,
\mathbf{F}_{N L_{\mathrm{P}}}^{\mathsf{H}}
\mathbf{E}_{L_{\mathrm{P}}}
\mathbf{F}_{N}$. Let $N_{+}=\left\lceil\frac{N}{2}\right\rceil$,
$N_{-}=\left\lfloor\frac{N}{2}\right\rfloor$. Then, 
\begin{equation}
\mathbf{E}_{L_{\mathrm{P}}}
=
\begin{bmatrix}
\mathbf{I}_{N_{+}}
&
\mathbf{0}_{N_{+}\times N_{-}}
\\
\mathbf{0}_{(N L_{\mathrm{P}}-N)\times N_{+}}
&
\mathbf{0}_{(N L_{\mathrm{P}}-N)\times N_{-}}
\\
\mathbf{0}_{N_{-}\times N_{+}}
&
\mathbf{I}_{N_{-}}
\end{bmatrix}
\in\mathbb{C}^{N L_{\mathrm{P}}\times N}.
\label{eq:ELP_def}
\end{equation}

It is well known that $L_{\mathrm{P}} \ge 4$ provides a good
discrete-time approximation of the continuous-time PAPR at the
transmit PA input \cite{ORPAPR,OFDMPAPR}. Therefore, in this paper,
we set $L_{\mathrm{P}}=4$ throughout.

\section{Problem Formulation}
\label{ProblemF}

In this section, we formulate the AFDM waveform design problem for ISAC under different design settings. Building on the signal and metric definitions in Section~\ref{Preliminaries}, we reserve a subset of chirp-subcarriers that do not carry data symbols and optimize the RCS symbols as a practical waveform-design handle. We also consider per-subcarrier pre-chirp parameters $c_{2,m}$ as additional design variables.

\subsection{Proposed AFDM Transmitter Structure}

\begin{figure*}[htbp]
 \centering
 \includegraphics[width=1.93\columnwidth,
 trim=00 1100 10 00,clip]{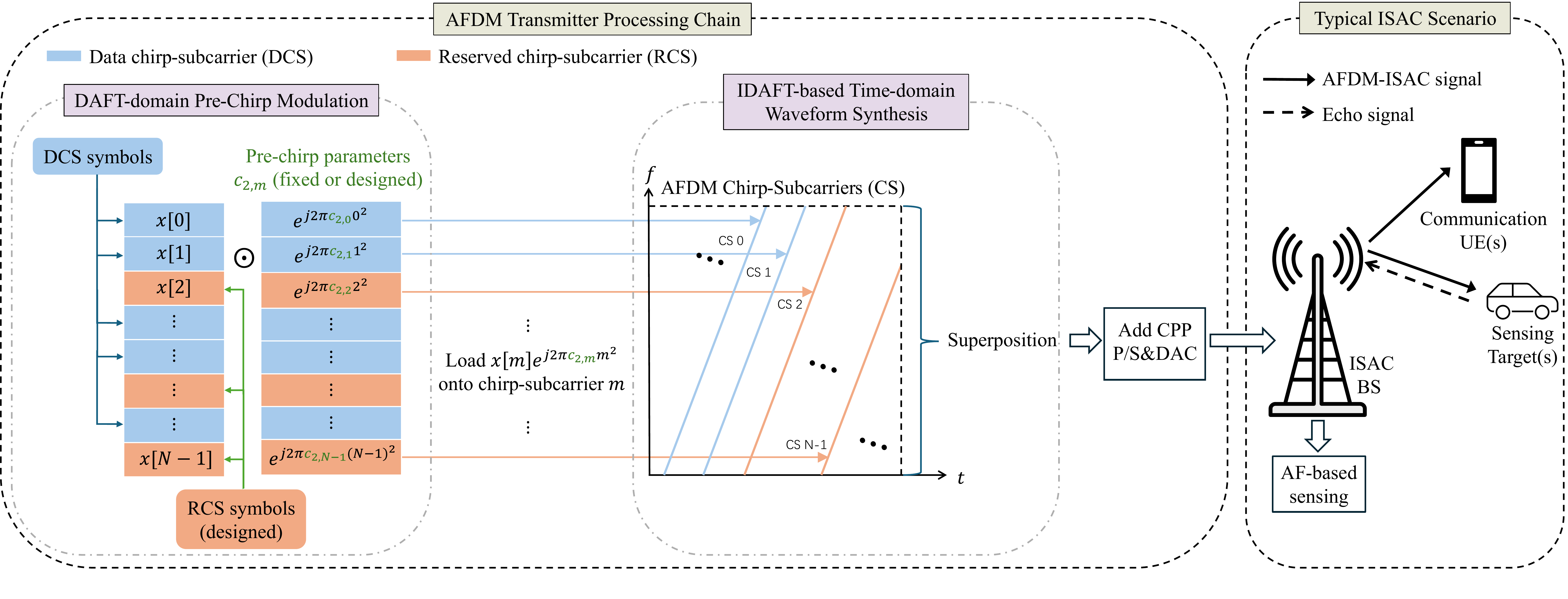}
\caption{Proposed AFDM transmitter processing chain and typical ISAC scenario.}
\vspace{-0.3cm}
 \label{Fig2}
\end{figure*}

We first partition the $N$ AFDM chirp-subcarriers into \emph{data chirp-subcarriers (DCSs)} and \emph{RCSs} (as shown in Fig.~{\ref{Fig2}}), indexed by $\mathcal{D}$ and $\mathcal{R}$, respectively, such that
\begin{align}
 \mathcal{D} \cup \mathcal{R} = \{0,1,\ldots,N-1\},
 \qquad
 \mathcal{D} \cap \mathcal{R} = \emptyset.
\end{align}
DCSs $\mathcal{D}$ are used to convey communication symbols, i.e. $x[m] \in \mathcal{X}, m \in \mathcal{D}$, where $\mathcal{X}$ denotes the chosen constellation (e.g., phase-shift keying (PSK) or quadrature amplitude modulation (QAM)). In contrast, RCSs
$\mathcal{R}$ do not carry user data and are instead dedicated to waveform design. \textcolor{black}{The partition $(\mathcal{D},\mathcal{R})$ is fixed for a given system configuration and is pre-agreed between the transmitter and receiver. It only needs to be communicated when the RCS allocation is updated, rather than for each AFDM symbol. For the highest-index RCS allocation used in the main simulations, the partition is fully specified by $|\mathcal{R}|$. For example, with $N=128$, at most $\lceil\log_2(N+1)\rceil=8$ bits are required to signal this count. Thus, this infrequent configuration signaling is not included in the effective spectral efficiency introduced later.}

\subsubsection{Reserved chirp-subcarrier symbols}

Since the RCSs do not carry information bits, the RCS symbols can be optimized directly as waveform-design variables, without requiring additional side information for symbol demapping at the receiver. Their amplitudes and phases may be jointly adjusted, subject to the total waveform-energy constraint introduced later.

\subsubsection{Per-subcarrier pre-chirp parameters}

To provide additional waveform-design flexibility, we allow the pre-chirp parameters on the DCSs to vary across subcarriers \cite{AFDM_PAPR}. For each DCS, the parameter $c_{2,m}$ is selected from a finite alphabet $\mathcal{C}_{2,m} = \left\{c_{2,m}^{(1)},\ldots,c_{2,m}^{(L_{\phi})}\right\}$. On an RCS, the factor $e^{j2\pi c_{2,m}m^2}$ can be absorbed into the designed RCS symbol. When the pre-chirp parameters are optimized, the selected indices require $B_{\mathrm{SI}} =|\mathcal{D}|\log_2 L_{\phi}$ side-information bits per AFDM symbol.

\textcolor{black}{To jointly account for the reduced number of data-carrying subcarriers due to the RCSs and the additional signaling overhead required by the optimized pre-chirp parameters, we define the effective spectral efficiency in \eqref{eq:eta_eff}, where $b_{\mathrm{data}}$ denotes the number of information bits carried by each data symbol, and $b_{\mathrm{SI}}$ denotes the number of bits carried by each side-information symbol:}
\begin{equation}
    \color{black}
    R_{\mathrm{eff}}
    =
    \frac{|\mathcal{D}|b_{\mathrm{data}}}
    {N+\left\lceil B_{\mathrm{SI}}/b_{\mathrm{SI}}\right\rceil}.
    \color{black}
    \label{eq:eta_eff}
\end{equation}
\textcolor{black}{The numerator in \eqref{eq:eta_eff} counts the information bits carried by the DCSs in one AFDM symbol, whereas the denominator counts the $N$ chirp-subcarriers in the original AFDM symbol together with the additional chirp-subcarriers required to transmit the side information. When mixed data constellations are used, $|\mathcal{D}|b_{\mathrm{data}}$ is interpreted as the total number of information bits carried by all DCSs. Thus, $R_{\mathrm{eff}}$ measures the effective number of information bits per chirp-subcarrier after accounting for both RCSs and side-information overhead, and its unit is bits/subcarrier.}

\subsubsection{Unified design vector}

Accordingly, in the proposed framework, the RCS symbols are always treated as waveform-design variables, while the per-subcarrier pre-chirp parameters on the DCSs may either be fixed or optimized depending on the design setting. To unify the subsequent formulations, we define $\mathbf{u} =[u[0],\ldots,u[N-1]]^{\mathsf{T}}$ as the actual DAFT-domain input vector to the IDAFT after pre-chirp modulation. For each DCS $m\in\mathcal{D}$,
\color{black}
\begin{equation}
u[m]
=
x[m]e^{j2\pi c_{2,m}m^{2}},
\qquad m\in\mathcal{D}.
\label{eq:u_DCS_def}
\end{equation}
\color{black}
When $c_{2,m}$ is optimized, the corresponding feasible set is
\begin{equation}
\mathcal{U}_m
=
\left\{
x[m]e^{j2\pi c_{2,m}m^{2}}
\,\big|\,
c_{2,m}\in\mathcal{C}_{2,m}
\right\},
\label{eq:Um_def}
\end{equation}
All elements of $\mathcal{U}_m$ have modulus $|x[m]|$. If $0\in\mathcal{D}$, we define $\mathcal{U}_0=\{x[0]\}$. 

For each RCS $m\in\mathcal{R}$, $u[m]\in\mathbb{C}$ is directly treated as a complex waveform-design variable. Under the RCS-only setting, the entries $\{u[m]\}_{m\in\mathcal{D}}$ are fixed.

With the above definition, the AFDM transmit waveform can be expressed
as
\color{black}
\begin{equation}
\mathbf{s}
=
\boldsymbol{\Phi}\mathbf{u},
\label{eq:s_Phi_u}
\end{equation}
\color{black}
where $\boldsymbol{\Phi}
=
\boldsymbol{\Lambda}_{c_1}
\mathbf{F}_{N}^{\mathsf{H}}$. Accordingly, the $L_{\mathrm{P}}$-times oversampled waveform defined
in \eqref{eq:s_tilde_matrix} becomes
\color{black}\begin{equation}
\tilde{\mathbf{s}}^{(\mathrm{P})}
=
\boldsymbol{\Phi}^{(\mathrm{P})}\mathbf{u},
\label{eq:sP_PhiP_u}
\end{equation}\color{black}
where $\boldsymbol{\Phi}^{(\mathrm{P})}
=
\mathbf{P}_{L_{\mathrm{P}}}\boldsymbol{\Phi}$.

The total DAFT-domain waveform energy is
constrained as $\|\mathbf{u}\|^{2}
=
E_{\mathrm{T}}$, where $E_{\mathrm{T}}>0$ is a prescribed total waveform energy.

\noindent\textit{Receiver-side processing:}
When per-subcarrier pre-chirp parameters are used and made available at the intended receiver, the receiver processing follows a conventional AFDM receiver with $c_{2,m}$ (instead of a common $c_2$) \cite{AFDMOR, AFDM_PAPR}, so the detailed implementation is omitted for brevity. If the pre-chirp parameters are fixed across the DCSs, the receiver reduces to the corresponding standard AFDM implementation.  After CP removal and DAFT processing, only the symbols on $\mathcal{D}$ are demapped, while those on $\mathcal{R}$ are discarded.

\subsection{Design Modes and Performance Objectives}
\label{subsec:design_modes}

Under either waveform-design configuration, the proposed framework supports the following three design modes, which reflect different priorities between sensing performance and hardware efficiency.

\emph{1) AF shaping mode:}
The first mode focuses on sensing performance by minimizing the weighted ISL. This mode can be implemented either by optimizing only the RCS symbols or by jointly optimizing the RCS symbols and the per-subcarrier pre-chirp parameters. The goal is to suppress AF sidelobes within certain LAZ, thereby improving target detectability. 

\emph{2) PAPR minimization mode:}
The second mode aims to minimize the waveform PAPR. Likewise, this mode can be realized either with RCS-symbol optimization only or with additional per-subcarrier pre-chirp design. This mode focuses on reducing the dynamic range at the PA input, which mitigates nonlinear distortion.

\emph{3) Joint AF shaping and PAPR control mode:}
The third mode jointly accounts for sensing and hardware constraints by minimizing the ISL subject to an explicit PAPR constraint, i.e., by enforcing $\mathrm{PAPR} \le \Gamma$ for a prescribed threshold $\Gamma$. This mode balances AF-sidelobe suppression in the LAZ and the PAPR requirement. This mode will be the primary focus of our subsequent development, while the first two modes can be considered as special cases of the unified formulation based on this mode.

\subsection{Unified Optimization Problem}

To cover both the RCS-only setting and the setting that further exploits per-subcarrier pre-chirp parameters within a unified framework, we introduce the following optimization problem:
\begin{equation}
\begin{aligned}
 \min_{\mathbf{u}}
 \quad &
 J_{\mathrm{ISL}}(\mathbf{u})
 \\ 
 \text{s.t.}\quad
 &
 \mathrm{PAPR}_{\mathbf{u}}
 \le \Gamma,
 \\ 
 &
\|\mathbf{u}\|^{2} = E_{\mathrm{T}},
 \\ 
 &
 u[m] \in \mathcal{U}_m,
 \quad
 m \in \mathcal{D},
\end{aligned}
\label{eq:unified_problem}
\end{equation}
where $J_{\mathrm{ISL}}(\mathbf{u})$ and $\mathrm{PAPR}_{\mathbf{u}}$ are specified next. In the analytical development, $\Gamma$ is used on a linear scale; when its value is reported in dB, it refers to $10\log_{10}\Gamma$.

\paragraph*{ISL in terms of $\mathbf{u}$}
Based on \eqref{eq:s_Phi_u}, the AF at $(\tau,\mu_{q})$ can be written as $A_{\tau,\mu_{q}} (\mathbf{u})
 =
 \mathbf{u}^{\mathsf{H}}
 \mathbf{C}_{\tau,\mu_{q}}
 \mathbf{u}$,
 $\mathbf{C}_{\tau,\mu_{q}}
 =
 \boldsymbol{\Phi}^{\mathsf{H}}
 \mathbf{U}_{\tau,\mu_{q}}
 \boldsymbol{\Phi}$.
Then, using the LAZ sampling set $\mathcal{A}$ and weights $w_{\tau,\mu_{q}}$, the weighted ISL becomes
\begin{equation}
 J_{\mathrm{ISL}}(\mathbf{u})
 =
 \sum_{(\tau,\mu_{q})\in\mathcal{A}}
 w_{\tau,\mu_{q}}\,
 \left|
 \mathbf{u}^{\mathsf{H}}
 \mathbf{C}_{\tau,\mu_{q}}
 \mathbf{u}
 \right|^{2}.
 \label{eq:ISL_u_def}
\end{equation}

\paragraph*{PAPR in terms of $\mathbf{u}$}

Based on \eqref{eq:sP_PhiP_u}, the power of the $n$-th oversampled time-domain sample is
\begin{equation}
p_n(\mathbf{u})
=
\mathbf{u}^{\mathsf{H}}
\mathbf{G}_n
\mathbf{u},
\label{eq:pn_def}
\end{equation}
where $\mathbf{G}_n
=
\left(
\boldsymbol{\Phi}^{(\mathrm{P})}
\right)^{\mathsf{H}}
\mathbf{e}_n\mathbf{e}_n^{\mathsf{H}}
\boldsymbol{\Phi}^{(\mathrm{P})}$.

Using
$\mathbf{P}_{L_{\mathrm{P}}}^{\mathsf{H}}
\mathbf{P}_{L_{\mathrm{P}}}
=L_{\mathrm{P}}\mathbf{I}_N$, together with
$\boldsymbol{\Phi}^{\mathsf{H}}\boldsymbol{\Phi}
=\mathbf{I}_N$
and $\|\mathbf{u}\|^2=E_{\mathrm{T}}$, this gives $\frac{1}{N L_{\mathrm{P}}}
\left\|
\tilde{\mathbf{s}}^{(\mathrm{P})}
\right\|^2
=
\frac{E_{\mathrm{T}}}{N}$. Therefore,
\begin{equation}
\mathrm{PAPR}_{\mathbf{u}}
=
\frac{N}{E_{\mathrm{T}}}
\max_{0\le n<N L_{\mathrm{P}}}
\mathbf{u}^{\mathsf{H}}
\mathbf{G}_n
\mathbf{u}.
\label{eq:PAPR_u_def}
\end{equation}

For the RCS-only setting, the pre-chirp parameters on the DCSs are fixed according to a prescribed AFDM configuration. Equivalently, the corresponding entries of $\mathbf{u}$ on $\mathcal{D}$ are fixed to known constants, denoted by $\bar{u}[m]$, for $m\in\mathcal{D}$. Accordingly, the RCS-only counterpart of the joint AF shaping and PAPR control problem in \eqref{eq:unified_problem} is
\begin{equation}
\begin{aligned}
 \min_{\mathbf{u}}
 \quad &
 J_{\mathrm{ISL}}(\mathbf{u})
 \\ 
 \text{s.t.}\quad
 &
 \mathrm{PAPR}_{\mathbf{u}}
 \le \Gamma,
 \\ 
 &
 u[m] = \bar{u}[m],
 \quad
 m \in \mathcal{D},
 \\ 
 &
\|\mathbf{u}\|^{2} = E_{\mathrm{T}}.
\end{aligned}
\label{eq:RS_only_problem}
\end{equation}
The formulation in \eqref{eq:unified_problem} and \eqref{eq:RS_only_problem} covers different design modes as special cases. The AF shaping mode is obtained by minimizing $J_{\mathrm{ISL}}(\mathbf{u})$ without the PAPR constraint, while the PAPR minimization mode is obtained by minimizing $\mathrm{PAPR}_{\mathbf{u}}$ and removing the PAPR constraint.

\vspace{-0.2cm}
\section{Proposed Algorithm}
\label{sec:Alg}

We develop a JIPD-MM algorithm for solving the unified AFDM waveform optimization problem in \eqref{eq:unified_problem}. The corresponding RCS-only designs are obtained by fixing the DCS entries as in \eqref{eq:RS_only_problem}.

The unified optimization problem in \eqref{eq:unified_problem} is challenging for three reasons: i) the weighted ISL objective $J_{\mathrm{ISL}}(\mathbf{u})$ in \eqref{eq:ISL_u_def} is nonconvex; ii) the PAPR metric in \eqref{eq:PAPR_u_def} contains a nonsmooth maximum over the oversampled sample powers; iii) the per-subcarrier sets $\mathcal{U}_m$ impose discrete-phase constraints. To cope with these difficulties, the proposed JIPD-MM algorithm proceeds in iterations indexed by $r$. At $r$-th iteration, the non-convex ISL objective and the PAPR constraint are replaced by smooth surrogate functions constructed around the current iterate $\mathbf{u}^{(r)}$. The construction of the ISL surrogate is presented in Section~\ref{subsec:ISL_surrogate}. In Section~\ref{subsec:PAPR_surrogate}, the max-type PAPR constraint is converted into a tractable smooth surrogate and incorporated via a penalty formulation. In Section~\ref{subsec:NSP_discrete}, the discrete-phase constraints are handled by an initialization stage based on convex-hull relaxation and negative square penalty (NSP), followed by a main optimization stage over the original alphabets. The overall proposed JIPD-MM algorithm is presented in Section~\ref{subsec:JIPD_overall}.

\subsection{ISL Surrogate}
\label{subsec:ISL_surrogate}

Based on the weighted ISL definition in \eqref{eq:ISL_u_def}, we now construct a linear MM surrogate for $J_{\mathrm{ISL}}(\mathbf{u})$ by repeatedly leveraging the following lemma~\cite{Song2016}.

\begin{lemma}[Quadratic majorization]
\label{lemma1}
Let $\mathbf{Y},\mathbf{Z}\in\mathbb{C}^{d\times d}$ be Hermitian matrices such that $\mathbf{Z} \succeq \mathbf{Y}$. Then, for any reference point $\mathbf{q}^{(r)}\in\mathbb{C}^{d}$, the following inequality holds for all $\mathbf{q}\in\mathbb{C}^{d}$:
\begin{equation}
 \mathbf{q}^{\mathsf{H}}\mathbf{Y}\mathbf{q}
 \le
 \mathbf{q}^{\mathsf{H}}\mathbf{Z}\mathbf{q}
 + 2\Re\!\left\{
 \mathbf{q}^{\mathsf{H}}(\mathbf{Y}-\mathbf{Z})\mathbf{q}^{(r)}
 \right\}
 +
 \left(\mathbf{q}^{(r)}\right)^{\mathsf{H}}
 (\mathbf{Z}-\mathbf{Y})
 \mathbf{q}^{(r)}.
 \label{eq:lemma1_ineq}
\end{equation}
Thus, the quadratic function $\mathbf{q}^{\mathsf{H}}\mathbf{Y}\mathbf{q}$ can be majorized at $\mathbf{q}^{(r)}$ by the right-hand side of \eqref{eq:lemma1_ineq}.
\end{lemma}
\color{black}
The difference between the right-hand side and the left-hand side
of \eqref{eq:lemma1_ineq} is
\begin{equation}
\left(
\mathbf{q}-\mathbf{q}^{(r)}
\right)^{\mathsf{H}}
\left(
\mathbf{Z}-\mathbf{Y}
\right)
\left(
\mathbf{q}-\mathbf{q}^{(r)}
\right)
\ge 0.
\label{eq:quadratic_majorization_gap}
\end{equation}
Therefore, the gap is zero at $\mathbf{q}=\mathbf{q}^{(r)}$. Hence, the original quadratic function and its majorizer have the same value and the same first-order directional derivative at $\mathbf{q}^{(r)}$.
\color{black}

To apply \textit{Lemma~\ref{lemma1}}, we rewrite $J_{\mathrm{ISL}}(\mathbf{u})$ as a quadratic form of $\mathbf{z} = \operatorname{vec}(\mathbf{u}\mathbf{u}^{\mathsf{H}})$. 
For each $(\tau,\mu_{q})\in\mathcal{A}$, we define $\mathbf{v}_{\tau,\mu_{q}}
 =
 \operatorname{vec}\!\left(
 \left( \mathbf{C}_{\tau,\mu_{q}}\right)^{\mathsf{H}}
 \right)$, so that $\mathbf{u}^{\mathsf{H}}
  \mathbf{C}_{\tau,\mu_{q}}
 \mathbf{u}
 =
 \mathbf{v}_{\tau,\mu_{q}}^{\mathsf{H}}\mathbf{z}$.
Then, the weighted ISL in \eqref{eq:ISL_u_def} can be expressed as
\begin{equation}
 J_{\mathrm{ISL}}(\mathbf{u})
 =
 \mathbf{z}^{\mathsf{H}}
 \mathbf{J}
 \mathbf{z},
 \label{eq:ISL_vec_form}
\end{equation}
where $\mathbf{J}
 =
 \sum_{(\tau,\mu_{q})\in\mathcal{A}}
 w_{\tau,\mu_{q}}\,
 \mathbf{v}_{\tau,\mu_{q}}
 \mathbf{v}_{\tau,\mu_{q}}^{\mathsf{H}}
 \in\mathbb{C}^{N^{2}\times N^{2}}$.

Using $\|\mathbf{u}\|^{2}=E_{\mathrm{T}}$, which implies $\|\mathbf{z}\|^{2}=E_{\mathrm{T}}^{2}$, and applying \textit{Lemma~\ref{lemma1}} to \eqref{eq:ISL_vec_form}, after ignoring constant terms, $J_{\mathrm{ISL}}(\mathbf{u})$ can be majorized at $\mathbf{u}^{(r)}$ by
\begin{equation}
 \widetilde{J}_{\mathrm{ISL}}(\mathbf{u}\,|\,\mathbf{u}^{(r)})
 =
 2\Re\!\left\{
 (\mathbf{z}^{(r)})^{\mathsf{H}}
 \mathbf{M}_{\mathrm{J}}
 \mathbf{z}
 \right\},
 \label{eq:ISL_surrogate_z}
\end{equation}
where $\mathbf{z}^{(r)}
 =
 \operatorname{vec}\!\left(\mathbf{u}^{(r)}{\mathbf{u}^{(r)}}^{\mathsf{H}}\right)$, $\mathbf{M}_{\mathrm{J}}
 =
 \mathbf{J} - \lambda_{\mathrm{J}}\mathbf{I}_{N^{2}}$, and $\lambda_{\mathrm{J}} \ge \lambda_{\max}(\mathbf{J})$.

By substituting $\mathbf{z} = \operatorname{vec}(\mathbf{u}\mathbf{u}^{\mathsf{H}})$ back into \eqref{eq:ISL_surrogate_z}, the surrogate can be transformed into
\begin{align}
 &\widetilde{J}_{\mathrm{ISL}}(\mathbf{u}\,|\,\mathbf{u}^{(r)}) \nonumber \\
 =
 &2
 \sum_{(\tau,\mu_{q})\in\mathcal{A}}
 w_{\tau,\mu_{q}}\,
 \Re\!\left\{
 \left(\zeta_{\tau,\mu_{q}}^{(r)}\right)^{*}
 \mathbf{u}^{\mathsf{H}}
  \mathbf{C}_{\tau,\mu_{q}}
 \mathbf{u}
 \right\}
 \;-\;
 2\lambda_{\mathrm{J}}
 \left|
 \mathbf{u}^{\mathsf{H}}\mathbf{u}^{(r)}
 \right|^{2},
 \label{eq:ISL_surrogate_u}
\end{align}
where $\zeta_{\tau,\mu_{q}}^{(r)} = \left(\mathbf{u}^{(r)}\right)^{\mathsf{H}}  \mathbf{C}_{\tau,\mu_{q}} \mathbf{u}^{(r)}$.

By regrouping the quadratic terms in \eqref{eq:ISL_surrogate_u}, it can be written as
\begin{equation}
 \widetilde{J}_{\mathrm{ISL}}(\mathbf{u}\,|\,\mathbf{u}^{(r)})
 =
 \mathbf{u}^{\mathsf{H}}
 \mathbf{Q}_0^{(r)}
 \mathbf{u}
 + \mathrm{const},
 \label{eq:ISL_Q0_form}
\end{equation}
where
\begin{align}
\mathbf{Q}_0^{(r)}
 &=\widetilde{\mathbf{Q}}^{(r)}
 + \left(\widetilde{\mathbf{Q}}^{(r)}\right)^{\mathsf{H}}
\nonumber\\
 \widetilde{\mathbf{Q}}^{(r)}
 &=
 \sum_{(\tau,\mu_{q})\in\mathcal{A}}
 w_{\tau,\mu_{q}}\,
 \left(\zeta_{\tau,\mu_{q}}^{(r)}\right)^{*}
  \mathbf{C}_{\tau,\mu_{q}}
 - \lambda_{\mathrm{J}}\,
 \mathbf{u}^{(r)}{\mathbf{u}^{(r)}}^{\mathsf{H}}.
\end{align}

Applying \textit{Lemma~\ref{lemma1}} again to \eqref{eq:ISL_Q0_form}, we obtain a linear surrogate
\begin{equation}
 \widehat{J}_{\mathrm{ISL}}(\mathbf{u}\,|\,\mathbf{u}^{(r)})
 =
 2\Re\!\left\{
 \left(\mathbf{d}^{(r)}\right)^{\mathsf{H}}\mathbf{u}
 \right\},
 \label{eq:ISL_surrogate_linear}
\end{equation}
where $\mathbf{d}^{(r)}
 =
 \left(\mathbf{M}_0^{(r)}\right)^{\mathsf{H}}\mathbf{u}^{(r)}$, $\mathbf{M}_0^{(r)}
 =
 \mathbf{Q}_0^{(r)} - \lambda_{\mathrm{Q}}^{(r)}\mathbf{I}_{N}$, and $\lambda_{\mathrm{Q}}^{(r)}
 \ge
 \lambda_{\max}\!\left(\mathbf{Q}_0^{(r)}\right)$.

\color{black}
The constant terms omitted from \eqref{eq:ISL_surrogate_z}--\eqref{eq:ISL_surrogate_linear} do not depend on $\mathbf{u}$ and hence do not affect the update. The corresponding full majorizer coincides with $J_{\mathrm{ISL}}(\mathbf{u})$ at $\mathbf{u}=\mathbf{u}^{(r)}$ and has the same first-order directional derivative there.
\color{black}

\subsection{PAPR Surrogate and Penalty Formulation}
\label{subsec:PAPR_surrogate}
In this subsection, we construct an MM-based surrogate and a penalty formulation for the PAPR constraint in \eqref{eq:unified_problem}. From \eqref{eq:PAPR_u_def}, the constraint $\mathrm{PAPR}_{\mathbf{u}}\le\Gamma$ is equivalent to
\begin{equation}
\max_{0\le n<N L_{\mathrm{P}}}
p_n(\mathbf{u})
\le
\Gamma_{\mathrm{p}},
\label{eq:PAPR_peak_constraint}
\end{equation}
where $\Gamma_{\mathrm{p}}
=
\frac{E_{\mathrm{T}}}{N}\Gamma$.

To handle the non-smooth max operator in \eqref{eq:PAPR_peak_constraint}, we consider an $\ell$-norm approximation. Specifically, we conservatively approximate \eqref{eq:PAPR_peak_constraint} by
\begin{equation}
 \left(
 \sum_{n=0}^{N L_{\mathrm{P}}-1}
 \left( p_n(\mathbf{u}) \right)^{\ell}
 \right)^{\!1/\ell}
 \le \Gamma_{\text{p}},
 \label{eq:PAPR_ell_norm}
\end{equation}
where $\ell \ge 2$ is an integer. \textcolor{black}{The left-hand side of \eqref{eq:PAPR_ell_norm} upper-bounds $\max_n p_n(\mathbf{u})$ and converges to it as $\ell \to \infty$. Since the average sample power is fixed at $E_{\mathrm{T}}/N$, controlling this smooth upper bound also controls the PAPR in \eqref{eq:PAPR_u_def}. Consequently, by choosing $\ell$ sufficiently large, the peak sample power constraint in \eqref{eq:PAPR_peak_constraint} can be conservatively approximated by
\begin{equation}
 \sum_{n=0}^{N L_{\mathrm{P}}-1}
 \left( p_n(\mathbf{u}) \right)^{\ell}
 \le \Gamma_{\ell},
 \label{eq:PAPR_moment_constraint}
\end{equation}
where $\Gamma_{\ell} = \Gamma_{\text{p}}^{\ell}$.}

Then, we derive an MM-compatible surrogate for \eqref{eq:PAPR_moment_constraint} using the following lemma~\cite{Song2016}.

\begin{lemma}[$\ell$-norm majorization]
\label{lemma:PAPR_power}
Consider the $\ell$-th power function $f(x)=x^{\ell}$ with $\ell\ge2$ and $x\in[0,t]$. For any $x_0\in[0,t)$, $f(x)$ is majorized at $x_0$ over the interval $[0,t]$ by the quadratic function
\begin{equation}
 \alpha x^2
 +
 \left(\ell x_0^{\ell-1} - 2\alpha x_0\right) x
 +
 \alpha x_0^2
 -
 (\ell-1) x_0^{\ell},
 \label{eq:lemma2_quad}
\end{equation}
where
\begin{equation}
 \alpha
 =
 \frac{t^{\ell} - x_0^{\ell} - \ell x_0^{\ell-1} (t-x_0)}
 {(t-x_0)^2}.
 \label{eq:lemma2_alpha}
\end{equation}
\end{lemma}
\textcolor{black}{The quadratic majorizer in \eqref{eq:lemma2_quad} has the same
value and the same first derivative as $f(x)=x^{\ell}$ at
$x=x_0$.}

At the $r$-th iteration, let $t_{\mathrm{P}}> \max_n p_n(\mathbf{u}^{(r)})$. Applying Lemma 2 to each term $\left(p_n(\mathbf{u})\right)^\ell$ over $[0,t_{\mathrm{P}}]$, we have
\begin{equation}
 \left( p_n(\mathbf{u}) \right)^{\ell}
 \le
 \alpha_n^{(r)} \left( p_n(\mathbf{u}) \right)^{2}
 + \beta_n^{(r)} p_n(\mathbf{u})
 + \Gamma_n^{(r)},
 \label{eq:PAPR_scalar_majorization}
\end{equation}
where
\begin{align}
 \alpha_n^{(r)}
 &=
 \frac{t_{\mathrm{P}}^{\ell}
 \!-\! \left(p_n\left(\mathbf{u}^{(r)}\right)\right)^{\ell}
 \!-\! \ell \left(p_n\left(\mathbf{u}^{(r)}\right)\right)^{\ell-1}
 \left(t_{\mathrm{P}}\!-\!p_n\left(\mathbf{u}^{(r)}\right)\right)}
 {\left(t_{\mathrm{P}}-p_n\left(\mathbf{u}^{(r)}\right)\right)^2},
 \nonumber\\ 
 \beta_n^{(r)}
 &= \ell \left(p_n\left(\mathbf{u}^{(r)}\right)\right)^{\ell-1}
 - 2 \alpha_n^{(r)} p_n\left(\mathbf{u}^{(r)}\right),
 \nonumber\\ 
 \Gamma_n^{(r)}
 &= \alpha_n^{(r)} \left(p_n\left(\mathbf{u}^{(r)}\right)\right)^2
 - (\ell-1)\left(p_n\left(\mathbf{u}^{(r)}\right)\right)^{\ell}.
\end{align}
In the numerical implementation, $t_{\mathrm{P}}$ is initialized as $1.1\max_n p_n(\mathbf{u}^{(r)})$. If the candidate update satisfies $\max_n p_n(\mathbf{u})>t_{\mathrm{P}}$, we set $t_{\mathrm{P}}\leftarrow1.1t_{\mathrm{P}}$ and recompute the current iteration until the candidate lies within $[0,t_{\mathrm{P}}]$.

Based on \eqref{eq:PAPR_scalar_majorization}, we have
\begin{align}
&\sum_{n=0}^{N L_{\mathrm{P}}-1}
\left(p_n(\mathbf{u})\right)^{\ell}
\nonumber\\
&\le
\sum_{n=0}^{N L_{\mathrm{P}}-1}
\left(
\alpha_n^{(r)}
\left(p_n(\mathbf{u})\right)^2
+
\beta_n^{(r)}p_n(\mathbf{u})
+
\Gamma_n^{(r)}
\right).
\label{eq:PAPR_moment_majorized}
\end{align}
Therefore, a sufficient surrogate condition for
\eqref{eq:PAPR_moment_constraint} is
\begin{equation}
\sum_{n=0}^{N L_{\mathrm{P}}-1}
\left(
\alpha_n^{(r)}
\left(p_n(\mathbf{u})\right)^2
+
\beta_n^{(r)}p_n(\mathbf{u})
\right)
\le
E_1^{(r)},
\label{eq:PAPR_quad_constraint}
\end{equation}
where $E_1^{(r)}
=
\Gamma_{\ell}
-
\sum_{n=0}^{N L_{\mathrm{P}}-1}
\Gamma_n^{(r)}$.

Since $\|\mathbf{u}\|^{2}=E_{\mathrm{T}}$, each term on the left-hand side of \eqref{eq:PAPR_quad_constraint} can be rewritten as
\begin{align}
 \alpha_n^{(r)} \!\left( p_n(\mathbf{u}) \right)^{2}
 \!+\! \beta_n^{(r)} p_n(\mathbf{u}) 
 \!=& \alpha_n^{(r)}\!
 \left(
 \!p_n(\mathbf{u})
 \!+\! \frac{\beta_n^{(r)}}{2\alpha_n^{(r)}}\!
 \right)^{\!2}
 \!-\! \frac{\left(\beta_n^{(r)}\right)^{2}}{4\alpha_n^{(r)}}
 \nonumber\\ 
 =& \alpha_n^{(r)}
 \left(
 \mathbf{u}^{\mathsf{H}}
 \widetilde{\mathbf{G}}_n^{(r)}
 \mathbf{u}
 \right)^{\!2}
 - \frac{\left(\beta_n^{(r)}\right)^{2}}{4\alpha_n^{(r)}},
 \label{eq:PAPR_complete_square}
\end{align}
where $\widetilde{\mathbf{G}}_n^{(r)}
=
\mathbf{G}_n
+
\frac{\beta_n^{(r)}}
{2E_{\mathrm{T}}\alpha_n^{(r)}}\mathbf{I}_N$. Substituting \eqref{eq:PAPR_complete_square} into \eqref{eq:PAPR_quad_constraint} and collecting constant terms yield
\begin{equation}
 \sum_{n=0}^{N L_{\mathrm{P}}-1}
 \alpha_n^{(r)}
 \left(
 \mathbf{u}^{\mathsf{H}}
 \widetilde{\mathbf{G}}_n^{(r)}
 \mathbf{u}
 \right)^{\!2}
 \le
 E_{2}^{(r)},
 \label{eq:PAPR_quad_constraint_u}
\end{equation}
where $E_{2}^{(r)}
 =
 E_{1}^{(r)}
 +
 \sum_{n=0}^{N L_{\mathrm{P}}-1}
 \frac{\left(\beta_n^{(r)}\right)^{2}}{4\alpha_n^{(r)}}$.

Then, the left-hand side of \eqref{eq:PAPR_quad_constraint_u} can be majorized in a manner similar to the handling of the ISL terms in
\eqref{eq:lemma1_ineq}--\eqref{eq:ISL_Q0_form}, as follows:
\begin{align}
\sum_{n=0}^{N L_{\mathrm{P}}-1}
\alpha_n^{(r)}
\left(
\mathbf{u}^{\mathsf{H}}
\widetilde{\mathbf{G}}_n^{(r)}
\mathbf{u}
\right)^{\!2}
&=
\mathbf{z}^{\mathsf{H}}
\mathbf{L}^{(r)}
\mathbf{z}
\nonumber\\
&\le
2\Re\!\left\{
\left(\mathbf{z}^{(r)}\right)^{\mathsf{H}}
\mathbf{M}_{\mathrm{L}}^{(r)}
\mathbf{z}
\right\}
+
C_{0}^{(r)},
\label{eq:PAPR_vec_majorization}
\end{align}
where
\begin{align}
\mathbf{L}^{(r)}
&=
\sum_{n=0}^{N L_{\mathrm{P}}-1}
\alpha_n^{(r)}
\mathbf{w}_n^{(r)}
\left(\mathbf{w}_n^{(r)}\right)^{\mathsf{H}},
\nonumber\\
\mathbf{M}_{\mathrm{L}}^{(r)}
&=
\mathbf{L}^{(r)}
-
\lambda_{\mathrm{L}}^{(r)}
\mathbf{I}_{N^{2}},
\nonumber\\
C_{0}^{(r)}
&=
\lambda_{\mathrm{L}}^{(r)}
E_{\mathrm{T}}^{2}
-
\left(\mathbf{z}^{(r)}\right)^{\mathsf{H}}
\mathbf{M}_{\mathrm{L}}^{(r)}
\mathbf{z}^{(r)},
\end{align}
with $\mathbf{w}_n^{(r)}
=
\operatorname{vec}\!\left(
\left(
\widetilde{\mathbf{G}}_n^{(r)}
\right)^{\mathsf{H}}
\right)$ and
$\lambda_{\mathrm{L}}^{(r)}
\ge
\lambda_{\max}\!\left(\mathbf{L}^{(r)}\right)$.

Rewriting \eqref{eq:PAPR_vec_majorization} in terms of
$\mathbf{u}$, we have
\begin{align}
&2\Re\!\left\{
\left(\mathbf{z}^{(r)}\right)^{\mathsf{H}}
\mathbf{M}_{\mathrm{L}}^{(r)}
\mathbf{z}
\right\}
+
C_{0}^{(r)}
 =
\mathbf{u}^{\mathsf{H}}
\mathbf{Q}_{\mathrm{P},0}^{(r)}
\mathbf{u}
+
C_{1}^{(r)},
\label{eq:PAPR_vec_majorization2}
\end{align}
where
\begin{align}
\mathbf{Q}_{\mathrm{P},0}^{(r)}
&=
\frac{1}{2}
\left(
\widetilde{\mathbf{Q}}_{\mathrm{P}}^{(r)}
+
\left(
\widetilde{\mathbf{Q}}_{\mathrm{P}}^{(r)}
\right)^{\mathsf{H}}
\right),
\nonumber\\
\widetilde{\mathbf{Q}}_{\mathrm{P}}^{(r)}
&=
\sum_{n=0}^{N L_{\mathrm{P}}-1}
\left(
2\alpha_n^{(r)}
p_n\left(\mathbf{u}^{(r)}\right)
+
\beta_n^{(r)}
\right)
\mathbf{G}_n
\nonumber\\
&\quad
-
2\lambda_{\mathrm{L}}^{(r)}
\mathbf{u}^{(r)}
\left(\mathbf{u}^{(r)}\right)^{\mathsf{H}},
\nonumber\\
C_{1}^{(r)}
&=
C_{0}^{(r)}
+
\sum_{n=0}^{N L_{\mathrm{P}}-1}
\beta_n^{(r)}
\left(
p_n\left(\mathbf{u}^{(r)}\right)
+
\frac{\beta_n^{(r)}}{2\alpha_n^{(r)}}
\right).
\end{align}

Based on \eqref{eq:PAPR_vec_majorization} and \eqref{eq:PAPR_vec_majorization2}, the surrogate constraint \eqref{eq:PAPR_quad_constraint_u} can be replaced by
\begin{equation}
\mathbf{u}^{\mathsf{H}}
\mathbf{Q}_{\mathrm{P},0}^{(r)}
\mathbf{u}
\le
E_{3}^{(r)},
\label{eq:PAPR_quad_constraint_ub}
\end{equation}
where $E_{3}^{(r)}
=
E_{2}^{(r)}
-
C_{1}^{(r)}$.

Next, we incorporate this surrogate constraint into the objective function of \eqref{eq:unified_problem} via a quadratic penalty formulation. The PAPR-related surrogate penalty is written as
\begin{equation}
\widetilde{J}_{\mathrm{PAPR}}
\left(
\mathbf{u}\,|\,\mathbf{u}^{(r)}
\right)
=
\rho
\left|
\mathbf{u}^{\mathsf{H}}
\mathbf{Q}_{\mathrm{P},1}^{(r)}
\mathbf{u}
\right|^{2},
\label{eq:PAPR_surrogate_penalty}
\end{equation}
where $\rho\ge0$ denotes the penalty weight and
\begin{equation}
\mathbf{Q}_{\mathrm{P},1}^{(r)}
=
\mathbf{Q}_{\mathrm{P},0}^{(r)}
-
\frac{E_{3}^{(r)}}{E_{\mathrm{T}}}
\mathbf{I}_{N}.
\label{eq:QP1_def}
\end{equation}
Combining
\eqref{eq:PAPR_moment_majorized}--\eqref{eq:QP1_def} gives
\begin{equation}
\sum_{n=0}^{N L_{\mathrm{P}}-1}
\left(p_n(\mathbf{u})\right)^{\ell}
-
\Gamma_{\ell}
\le
\mathbf{u}^{\mathsf{H}}
\mathbf{Q}_{\mathrm{P},1}^{(r)}
\mathbf{u},
\label{eq:PAPR_residual_majorization}
\end{equation}
for any $\mathbf{u}$ satisfying $0\le p_n(\mathbf{u})\le t_{\mathrm{P}}$ for all $n$. \textcolor{black}{Equality holds at $\mathbf{u}=\mathbf{u}^{(r)}$, and the two sides have the same first-order directional derivative at this point.}

Finally, by re-applying \textit{Lemma~\ref{lemma1}} twice in similar manner as for the ISL term in \eqref{eq:ISL_vec_form}-\eqref{eq:ISL_surrogate_linear}, the penalty term in \eqref{eq:PAPR_surrogate_penalty} can be further majorized by the following linear surrogate
\begin{equation}
 \widehat{J}_{\mathrm{PAPR}}(\mathbf{u}\,|\,\mathbf{u}^{(r)})
 =
 2\rho\Re\!\left\{
 \left(\mathbf{c}^{(r)}\right)^{\mathsf{H}}\mathbf{u}
 \right\},
 \label{eq:PAPR_surrogate_linear}
\end{equation}
where
\begin{align*}
 \mathbf{c}^{(r)}
 &=
 \left(\mathbf{M}_{\mathrm{P},2}^{(r)}\right)^{\mathsf{H}}\mathbf{u}^{(r)}, \quad 
 \mathbf{M}_{\mathrm{P},2}^{(r)}
 =
 \mathbf{Q}_{\mathrm{P},2}^{(r)}
 - \lambda_{\mathrm{P},3}^{(r)}\mathbf{I}_{N},
 \\
 \lambda_{\mathrm{P},3}^{(r)}
 &\ge
 \lambda_{\max}\left(\mathbf{Q}_{\mathrm{P},2}^{(r)}\right),\quad 
 \mathbf{Q}_{\mathrm{P},2}^{(r)}
 =
 \frac{1}{2}\left(
 \widetilde{\mathbf{Q}}_{\mathrm{P},2}^{(r)}
 +
 \left(\widetilde{\mathbf{Q}}_{\mathrm{P},2}^{(r)}\right)^{\mathsf{H}}
 \right), \\
 \widetilde{\mathbf{Q}}_{\mathrm{P},2}^{(r)}
 &=
 2 \left(\mathbf{u}^{(r)}\right)^{\mathsf{H}}
 \mathbf{Q}_{\mathrm{P},1}^{(r)}
 \mathbf{u}^{(r)}
 \mathbf{Q}_{\mathrm{P},1}^{(r)}
 - 2 \lambda_{\mathrm{P},2}^{(r)}
 \mathbf{u}^{(r)}{\mathbf{u}^{(r)}}^{\mathsf{H}},\\
 \lambda_{\mathrm{P},2}^{(r)}
 &\ge \lambda_{\max}\left(\mathbf{L}_{\mathrm{P}}^{(r)}\right),\quad
 \mathbf{L}_{\mathrm{P}}^{(r)}
 =
 \mathbf{w}_{\mathrm{P}}^{(r)}
 \left(\mathbf{w}_{\mathrm{P}}^{(r)}\right)^{\mathsf{H}},\\
 \mathbf{w}_{\mathrm{P}}^{(r)}
 &=
 \operatorname{vec}\!\left(
 (\mathbf{Q}_{\mathrm{P},1}^{(r)})^{\mathsf{H}}
 \right).
\end{align*}
Based on the penalty term in \eqref{eq:PAPR_surrogate_linear}, an appropriate choice of the weight $\rho$ enables effective PAPR control within the JIPD-MM update.

\subsection{Discrete-Phase Constraint Handling}
\label{subsec:NSP_discrete}

In this subsection, we address the discrete-phase constraints $u[m]\in\mathcal{U}_m$ on the DCSs $m\in\mathcal{D}$ using a two-stage implementation. Inspired by the scheme for one-bit precoding in~\cite{shao2019}, the initialization stage employs a convex-hull relaxation, an NSP term, and the associated projection/normalization update to gradually steer the relaxed DCS variables toward the discrete-phase set and provide a favorable starting point for the subsequent main optimization stage. For the RCS-only designs, the initialization stage is skipped.

At the $r$-th iteration, the ISL and PAPR treatments in Sections~\ref{subsec:ISL_surrogate} and \ref{subsec:PAPR_surrogate} provide the linear surrogates $\widehat{J}_{\mathrm{ISL}}
(\mathbf{u}\,|\,\mathbf{u}^{(r)})$ in
\eqref{eq:ISL_surrogate_linear} and
$\widehat{J}_{\mathrm{PAPR}}
(\mathbf{u}\,|\,\mathbf{u}^{(r)})$ in \eqref{eq:PAPR_surrogate_linear}, respectively. The resulting discrete linear update is
\begin{equation}
\begin{aligned}
 \min_{\mathbf{u}}\quad &
 2\Re\!\left\{
 \left(\mathbf{d}^{(r)}+\rho\mathbf{c}^{(r)}\right)^{\mathsf{H}}
 \mathbf{u}
 \right\}
 \\
 \text{s.t.}\quad &
 u[m]\in\mathcal{U}_m,\quad m\in\mathcal{D},\\
 & \|\mathbf{u}\|^{2}=E_{\mathrm{T}}.
\end{aligned}
\label{eq:MM_update_discrete}
\end{equation}

\textit{Initialization stage:}
In the initialization stage, each discrete alphabet
$\mathcal{U}_m$ is relaxed to its convex hull
$\overline{\mathcal{U}}_m
=\operatorname{conv}(\mathcal{U}_m)$, and the following NSP
term is introduced:
\begin{equation}
\begin{aligned}
 \min_{\mathbf{u}}\quad &
 2\Re\!\left\{
 \left(\mathbf{d}^{(r)}+\rho\mathbf{c}^{(r)}\right)^{\mathsf{H}}
 \mathbf{u}
 \right\}
 -\omega\sum_{m\in\mathcal{D}}|u[m]|^{2}
 \\
 \text{s.t.}\quad &
 u[m]\in\overline{\mathcal{U}}_m,\quad m\in\mathcal{D},\\
 & \|\mathbf{u}\|^{2}=E_{\mathrm{T}},
\end{aligned}
\label{eq:MM_update_NSP_relaxed}
\end{equation}
where $\omega>0$ is the NSP weight.

Each $\overline{\mathcal{U}}_m$ is a convex polygon in the complex plane whose vertices coincide with $\mathcal{U}_m$. The concave quadratic term $-\omega\sum_{m\in\mathcal{D}}|u[m]|^{2}$ encourages the solution to approach the vertices of these polygons. Its first-order upper bound at $\mathbf{u}^{(r)}$ is
\begin{equation}
-\omega\sum_{m\in\mathcal{D}}|u[m]|^{2}
\le
2\omega\Re\!\left\{
\left(\mathbf{g}_{\mathrm{NSP}}^{(r)}\right)^{\mathsf{H}}
\mathbf{u}
\right\}
+\mathrm{const}_{\mathrm{NSP}}^{(r)},
\label{eq:NSP_linear_majorant}
\end{equation}
where
$\mathrm{const}_{\mathrm{NSP}}^{(r)}
$ is the constant term and
\begin{equation}
g_{\mathrm{NSP}}^{(r)}[m]
=
\begin{cases}
-\,u^{(r)}[m], & m\in\mathcal{D},\\[0.2em]
0, & m\in\mathcal{R}.
\end{cases}
\label{eq:g_NSP_def}
\end{equation}
The bound in \eqref{eq:NSP_linear_majorant} is tight at $\mathbf{u}=\mathbf{u}^{(r)}$.

Combining \eqref{eq:MM_update_NSP_relaxed} and \eqref{eq:NSP_linear_majorant}, and ignoring terms independent
of $\mathbf{u}$, the overall NSP-relaxed update at $r$-th iteration reduces to
\begin{equation}
\begin{aligned}
\min_{\mathbf{u}}\quad &
2\Re\!\left\{
\left(\tilde{\mathbf{g}}^{(r)}\right)^{\mathsf{H}}\mathbf{u}
\right\}\\
\text{s.t.}\quad &
u[m]\in\overline{\mathcal{U}}_m,\quad m\in\mathcal{D},\\
& \|\mathbf{u}\|^{2}=E_{\mathrm{T}},
\end{aligned}
\label{eq:MM_NSP_linear_prob}
\end{equation}
where $\tilde{\mathbf{g}}^{(r)}
=
\mathbf{d}^{(r)}
+\rho\mathbf{c}^{(r)}
+\omega\mathbf{g}_{\mathrm{NSP}}^{(r)}$. By completing the square,
\eqref{eq:MM_NSP_linear_prob} is equivalent to
\begin{equation}
\begin{aligned}
\min_{\mathbf{u}}\quad &
\left\|
\mathbf{u}-\breve{\mathbf{u}}^{(r)}
\right\|^{2}\\
\text{s.t.}\quad &
u[m]\in\overline{\mathcal{U}}_m,\quad m\in\mathcal{D},\\
& \|\mathbf{u}\|^{2}=E_{\mathrm{T}},
\end{aligned}
\label{eq:prox_subproblem}
\end{equation}
where
$\breve{\mathbf{u}}^{(r)}
=-\tilde{\mathbf{g}}^{(r)}$.

The total-energy constraint couples the DCS and RCS blocks in \eqref{eq:prox_subproblem}. In the initialization stage, we use
the following block-wise projection/normalization update. For each $m\in\mathcal{D}$, define
\[
\tilde{u}^{(r)}[m]
=
\Pi_{\overline{\mathcal{U}}_m}
\left(\breve{u}^{(r)}[m]\right).
\]
The update is
\begin{equation}
u^{(r+1)}[m]
=
\begin{cases}
\tilde{u}^{(r)}[m], & m\in\mathcal{D},\\[0.3em]
\varepsilon^{(r)}\breve{u}^{(r)}[m],
& m\in\mathcal{R},
\end{cases}
\label{eq:u_update_projection}
\end{equation}
where
\begin{equation}
\varepsilon^{(r)}
=
\sqrt{
\frac{
E_{\mathrm{T}}
-\displaystyle\sum_{m\in\mathcal{D}}
|\tilde{u}^{(r)}[m]|^{2}
}{
\displaystyle\sum_{m\in\mathcal{R}}
|\breve{u}^{(r)}[m]|^{2}
}
}.
\label{eq:eta_def}
\end{equation}
Fig.~\ref{fig:NSP_projection} illustrates the projection step for a single DCS. For a given index $m\in\mathcal{D}$, the red dots on the unit circle represent a discrete eight-point example alphabet $\mathcal{U}_m$, and the black polygon is its convex hull $\overline{\mathcal{U}}_m$. The black circles denote several possible pre-projection values $\breve{u}^{(r)}[m]$ produced by the MM linear surrogate, whereas the black crosses indicate the projected values $\tilde{u}^{(r)}[m] = \Pi_{\overline{\mathcal{U}}_m}\left(\breve{u}^{(r)}[m]\right)$. If $\breve{u}^{(r)}[m]\notin\overline{\mathcal{U}}_m$, the projection lands on an edge or a vertex of the polygon, while if $\breve{u}^{(r)}[m]\in\overline{\mathcal{U}}_m$, the projection coincides with the pre-projection point itself. The shaded sector and dashed rays highlight how projections in a representative angular region are treated. The same convex-hull projection principle applies to any discrete alphabet $\mathcal{U}_m$. For the example in Fig.~\ref{fig:NSP_projection} and in our simulations, we use an octagonal feasible set of the form $\mathcal{U}_m = \left\{ x[m]e^{\jmath \varphi_\ell}: \varphi_\ell = \varphi_0 + \delta + \ell \tfrac{\pi}{4}, \ \ell = 0,\ldots,7 \right\}$, where $\varphi_0$ is a fixed offset and $|\delta|\ll 1$ is a very small irrational rotation angle. For $m\neq0$, this feasible set corresponds to the pre-chirp-parameter alphabet $\mathcal{C}_{2,m} = \left\{ c_{2,m}^{(\ell)} : c_{2,m}^{(\ell)} = \frac{\varphi_0 + \delta + \ell \pi/4}{2\pi m^2}, \ \ell = 0,\ldots,7 \right\}$. By choosing $\delta$ appropriately, all $c_{2,m}^{(\ell)}$ can be made irrational, which is consistent with the AFDM requirement on the pre-chirp parameter in~\cite{AFDMOR}.

\begin{figure}[t]
\centering
\includegraphics[width=0.78\columnwidth]{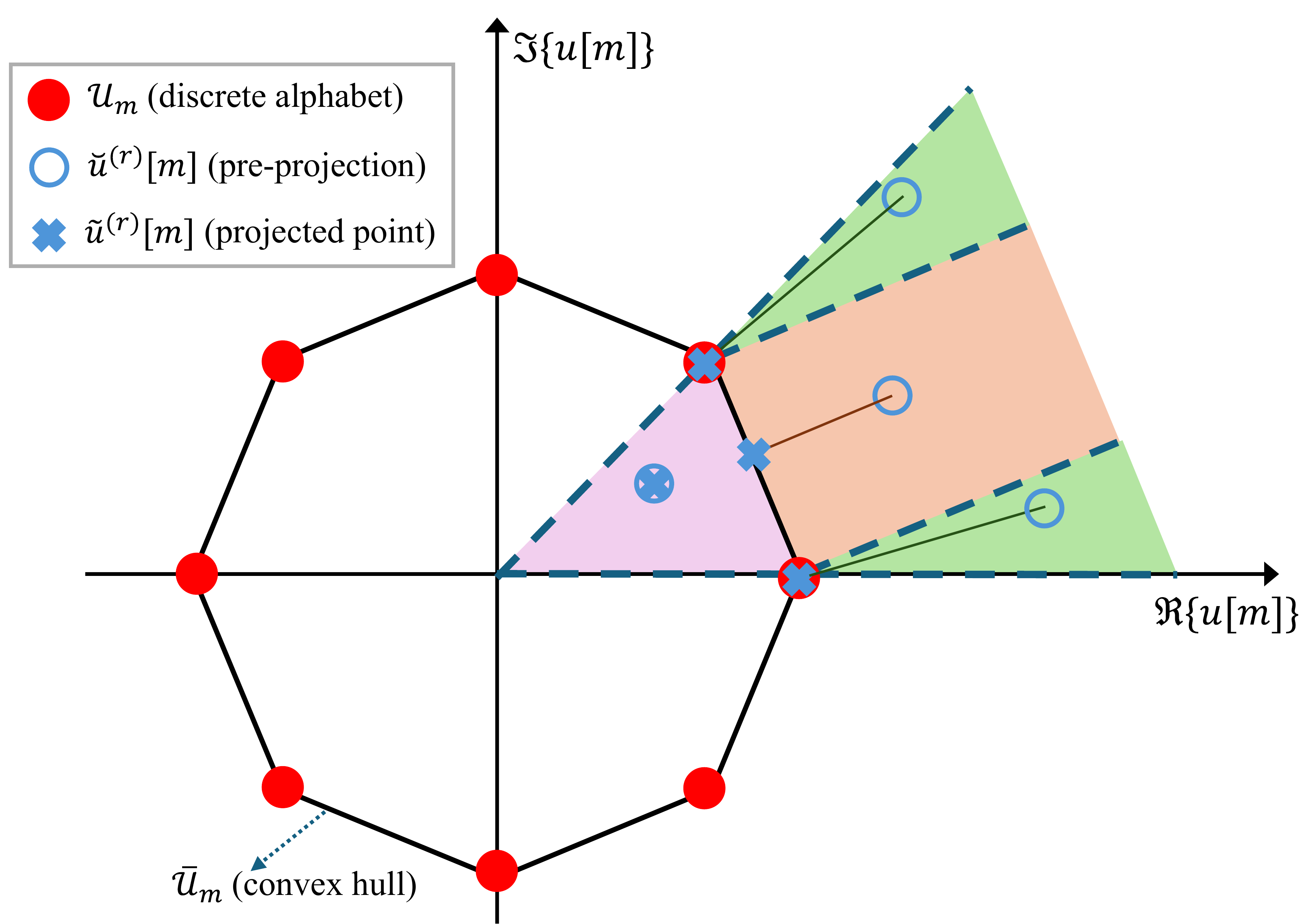}
\caption{Illustration of the convex-hull relaxation and per-subcarrier projection in the initialization stage for a DCS $m$ with discrete alphabet $\mathcal{U}_m$.}
\vspace{-0.3cm}
\label{fig:NSP_projection}
\end{figure}

The initialization stage is terminated after at most $r_{\mathrm{I}}^{\max}$ iterations. In the implementation of the joint AF shaping and PAPR control mode, it is terminated earlier if the PAPR constraint is satisfied for three consecutive iterations. The actual number of initialization iterations is denoted by $r_{\mathrm{I}}\le r_{\mathrm{I}}^{\max}$.

\textit{Main optimization stage:}
In this stage, the convex-hull relaxation and the NSP term are removed, and the original discrete MM subproblem in
\eqref{eq:MM_update_discrete} is solved exactly.

For each optimized DCS,
\begin{equation}
u^{(r+1)}[m]
=
\arg\min_{v\in\mathcal{U}_m}
\Re\left\{
\left(
d^{(r)}[m]+\rho c^{(r)}[m]
\right)^*v
\right\},
\label{eq:exact_DCS_update}
\end{equation}
for $m\in\mathcal{D}$. Since the DCS entries $u[m]$, $m\in\mathcal{D}$, have fixed moduli, the RCS entries are jointly updated under the remaining-energy constraint as
\begin{equation}
u^{(r+1)}[m]
\!=
\!-\!
\sqrt{
E_{\mathrm{T}}
\!-\!
\sum_{k\in\mathcal{D}}|x[k]|^2
}
\,
\frac{
d^{(r)}[m]\!+\!\rho c^{(r)}[m]
}{
\sqrt{
\displaystyle
\sum_{k\in\mathcal{R}}
\left|
d^{(r)}[k]\!+\!\rho c^{(r)}[k]
\right|^2
}
},
\label{eq:exact_RCS_update}
\end{equation}
for $m\in\mathcal{R}$. For the RCS-only designs, the DCS entries remain fixed and only the RCS update in \eqref{eq:exact_RCS_update} is required.

\subsection{Overall JIPD-MM Iteration}
\label{subsec:JIPD_overall}

In the previous subsections, MM surrogates have been constructed for the ISL objective and the PAPR penalty in \eqref{eq:ISL_surrogate_linear} and \eqref{eq:PAPR_surrogate_linear}, respectively. The discrete-phase constraints are handled by the initialization and main optimization stages described in
Section~\ref{subsec:NSP_discrete}. In this subsection, we assemble these components into the overall JIPD-MM algorithm and discuss convergence and per-iteration computational complexity.

\subsubsection{JIPD-MM Algorithm}

The overall JIPD-MM procedure is summarized in \textbf{Algorithm~\ref{alg:JIPD-MM}}. When the per-subcarrier pre-chirp parameters on the DCSs are optimized, the initialization stage applies the projection/normalization update. The resulting point is then mapped once to the discrete feasible set, after which the
main optimization stage applies \eqref{eq:exact_DCS_update} and \eqref{eq:exact_RCS_update} until termination. For the RCS-only designs, the initialization stage is skipped. In the numerical implementation of the initialization stage, $\rho$ is adjusted according to the achieved PAPR, while the NSP weight $\omega$ is gradually increased. After the transition, $\rho$ is kept fixed during the main optimization stage.

\begin{algorithm}[t]
 \caption{JIPD-MM AFDM waveform design}
 \label{alg:JIPD-MM}
 \begin{algorithmic}[1]
\Statex \textbf{Input:}
$\mathbf{u}^{(0)}$, $\mathcal{D},\mathcal{R}, \{\mathcal{U}_m\}$, $E_{\mathrm{T}}$, $N$, $\Gamma$, $r_{\mathrm{I}}^{\max}$, $r_{\max}$

\Statex Set $r_{\mathrm{I}}=0$ for the RCS-only design; otherwise, set $r_{\mathrm{I}}=r_{\mathrm{I}}^{\max}$

\Statex \textit{Initialization stage (if $c_{2,m}$, $m \in \mathcal{D}$ are optimized)}
\For{$r=0,1,\ldots,r_{\mathrm{I}}^{\max}-1$}
     \State Form $\mathbf{g}_{\mathrm{cont}}^{(r)}
     =
     \mathbf{d}^{(r)}+\rho\mathbf{c}^{(r)}$ via the surrogates in \eqref{eq:ISL_surrogate_linear} and \eqref{eq:PAPR_surrogate_linear}, and form
     $\mathbf{g}_{\mathrm{NSP}}^{(r)}$ according to \eqref{eq:g_NSP_def}
     
     \State
     $\breve{\mathbf{u}}^{(r)}
     =
     -\mathbf{g}_{\mathrm{cont}}^{(r)}
     -\omega\mathbf{g}_{\mathrm{NSP}}^{(r)}$
     
     \State Update $\mathbf{u}^{(r+1)}$ from $\breve{\mathbf{u}}^{(r)}$ according to the projection and normalization rule in \eqref{eq:u_update_projection}--\eqref{eq:eta_def}
     \If{the transition condition is satisfied}
    \State Set $r_{\mathrm{I}}=r+1$ and \textbf{break}
\EndIf
 \State \textbf{end}
 \EndFor
 \State \textbf{end}

 \State Quantize each $u^{(r_{\mathrm{I}})}[m]$, $m\in\mathcal{D}$, to its nearest point in $\mathcal{U}_m$, and normalize the RCS entries accordingly

 \Statex \textit{Main optimization stage}
 \For{$r=r_{\mathrm{I}},\ldots,r_{\max}-1$}
     \State Form $\mathbf{d}^{(r)}$ and $\mathbf{c}^{(r)}$ via the ISL and PAPR surrogates in \eqref{eq:ISL_surrogate_linear} and \eqref{eq:PAPR_surrogate_linear}
     
     \State Update $\{u^{(r+1)}[m]\}_{m\in\mathcal{D}}$ according to \eqref{eq:exact_DCS_update}, if applicable

\State Update $\{u^{(r+1)}[m]\}_{m\in\mathcal{R}}$ according to \eqref{eq:exact_RCS_update}
     
     \If{the stopping criterion is met}
         \State \textbf{break}
     \EndIf
      \State \textbf{end}
 \EndFor
 \State \textbf{end}
 \Statex \textbf{Output:} Optimized AFDM waveform
 \end{algorithmic}
\end{algorithm}


\noindent\textbf{Remark 1:} \textbf{Algorithm~\ref{alg:JIPD-MM}} is presented for the joint AF shaping and PAPR control mode. The AF shaping mode is obtained by omitting the PAPR component and skipping the PAPR-related steps. The PAPR minimization mode is obtained by omitting the ISL component, setting $\Gamma=0$~dB, and keeping only the PAPR-related update. For the RCS-only designs, the DCS entries remain unchanged and the initialization stage is skipped.

{\color{black}
\subsubsection{Convergence Analysis}

The following monotonicity analysis concerns the two single-objective modes. When the per-subcarrier pre-chirp parameters are optimized, the initialization stage is performed for a finite number of iterations to provide a starting point for the subsequent main optimization stage; the RCS-only designs skip this stage. We therefore analyze the main-optimization-stage iterations. Both modes involve a single design criterion, and hence no trade-off weight affects the corresponding update. In the joint mode, the ISL and PAPR terms are optimized together and need not decrease individually; its averaged iteration behavior is examined numerically in Section~\ref{subsec:iteration_behavior}.

\begin{proposition}
\label{prop:convergence}
During the main optimization stage, the ISL objective in the AF shaping mode is nonincreasing and convergent. In the PAPR minimization mode, define
\begin{equation}
\mathrm{PAPR}_{\ell}(\mathbf{u})
\triangleq
\frac{N}{E_{\mathrm{T}}}
\left(
\sum_{n=0}^{N L_{\mathrm{P}}-1}
\left(p_n(\mathbf{u})\right)^{\ell}
\right)^{1/\ell}.
\label{eq:smooth_PAPR_def}
\end{equation}
Then, $\mathrm{PAPR}_{\ell}(\mathbf{u})$ upper-bounds $\mathrm{PAPR}_{\mathbf{u}}$ and satisfies $\lim_{\ell\rightarrow\infty}
\mathrm{PAPR}_{\ell}(\mathbf{u})
=
\mathrm{PAPR}_{\mathbf{u}}$. Moreover, $\mathrm{PAPR}_{\ell}(\mathbf{u}^{(r)})$ is nonincreasing and convergent.
\end{proposition}

\begin{IEEEproof}
For the AF shaping mode, let $\kappa_{\mathrm{ISL}}^{(r)}$
denote the terms independent of $\mathbf{u}$ that are omitted
in \eqref{eq:ISL_surrogate_linear}. Then,
$\widehat{J}_{\mathrm{ISL}}
(\mathbf{u}\mid\mathbf{u}^{(r)})
+\kappa_{\mathrm{ISL}}^{(r)}$
majorizes $J_{\mathrm{ISL}}(\mathbf{u})$ and satisfies equality at
$\mathbf{u}=\mathbf{u}^{(r)}$. The updates in \eqref{eq:exact_DCS_update} and \eqref{eq:exact_RCS_update} minimize the current linear surrogate over the feasible set; for the RCS-only designs, the DCS entries are fixed and \eqref{eq:exact_RCS_update} gives the corresponding minimizer over the RCS block. Hence,
\begin{align}
J_{\mathrm{ISL}}(\mathbf{u}^{(r+1)})
&\le
\widehat{J}_{\mathrm{ISL}}
(\mathbf{u}^{(r+1)}\mid\mathbf{u}^{(r)})
+\kappa_{\mathrm{ISL}}^{(r)}
\nonumber\\
&\le
\widehat{J}_{\mathrm{ISL}}
(\mathbf{u}^{(r)}\mid\mathbf{u}^{(r)})
+\kappa_{\mathrm{ISL}}^{(r)}
=
J_{\mathrm{ISL}}(\mathbf{u}^{(r)}).
\label{eq:ISL_descent}
\end{align}
Since $J_{\mathrm{ISL}}(\mathbf{u})\ge0$, its values along the main-optimization-stage iterations converge.

For the PAPR minimization mode, $\Gamma=0$~dB gives $\Gamma_{\mathrm{p}}=E_{\mathrm{T}}/N$. Since the average sample power satisfies
$\frac{1}{NL_{\mathrm{P}}}\sum_{n=0}^{NL_{\mathrm{P}}-1}
p_n(\mathbf{u})=\Gamma_{\mathrm{p}}$,
Jensen's inequality gives
$\sum_{n=0}^{NL_{\mathrm{P}}-1}(p_n(\mathbf{u}))^\ell
\ge NL_{\mathrm{P}}\Gamma_{\mathrm{p}}^\ell
\ge \Gamma_{\mathrm{p}}^\ell=\Gamma_\ell$. Under the update rule for $t_{\mathrm{P}}$ described after \eqref{eq:PAPR_scalar_majorization}, the majorizations leading to \eqref{eq:PAPR_surrogate_linear} provide a tight upper bound on the squared residual associated with \eqref{eq:PAPR_moment_constraint}. Since the main-optimization-stage update minimizes the current linear surrogate over the original feasible set,
\begin{align}
&
\left(
\sum_{n=0}^{N L_{\mathrm{P}}-1}
\left(p_n(\mathbf{u}^{(r+1)})\right)^{\ell}
-\Gamma_{\ell}
\right)^2
\nonumber\\
&\quad\le
\left(
\sum_{n=0}^{N L_{\mathrm{P}}-1}
\left(p_n(\mathbf{u}^{(r)})\right)^{\ell}
-\Gamma_{\ell}
\right)^2.
\label{eq:PAPR_descent}
\end{align}
Both residuals in \eqref{eq:PAPR_descent} are nonnegative. It then follows from \eqref{eq:smooth_PAPR_def} that
\[
\mathrm{PAPR}_{\ell}(\mathbf{u}^{(r+1)})
\le
\mathrm{PAPR}_{\ell}(\mathbf{u}^{(r)}).
\]
Since $\mathrm{PAPR}_{\ell}(\mathbf{u})$ is lower bounded, its values along the main-optimization-stage iterations converge.
\end{IEEEproof}
}

\subsubsection{Computational Complexity of the Proposed JIPD-MM Algorithm}

\textcolor{black}{Under the unified design-vector formulation, the
DCS entries are defined as $u[m] = x[m]e^{j2\pi c_{2,m}m^2}$.
Therefore, the system-dependent matrices
$\mathbf{C}_{\tau,\mu_q}$, $\mathbf{G}_n$, $\boldsymbol{\Phi}$,
and $\boldsymbol{\Phi}^{(\mathrm{P})}$ are fixed for a given
system configuration and do not need to be recomputed in the iterations.} Let $|\mathcal{A}|$ denote the number of delay--Doppler sampling points in the LAZ. For the ISL surrogate, each iteration requires evaluating
$\zeta^{(r)}_{\tau,\mu_q}$ and the corresponding products
needed to form $\mathbf{d}^{(r)}$ for all
$(\tau,\mu_q)\in\mathcal{A}$, which leads to a complexity of
order $\mathcal{O}(|\mathcal{A}|N\log(N))$ with the help of
fast Fourier transform (FFT). For the PAPR surrogate, generating the oversampled waveform and computing the sample powers dominate the cost and require $\mathcal{O}(L_{\mathrm{P}}N^2)$ operations. The $N^2\times N^2$ matrices
$\mathbf{L}^{(r)}$ and $\mathbf{L}_{\mathrm{P}}^{(r)}$
introduced in the derivation do not need to be explicitly
formed. Since $\mathbf{G}_n$ in \eqref{eq:pn_def} is rank one,
the products involving
$\mathbf{Q}_{\mathrm{P},0}^{(r)}$ and
$\mathbf{Q}_{\mathrm{P},2}^{(r)}$ that are required to obtain
$\mathbf{c}^{(r)}$ can be evaluated directly using
$\boldsymbol{\Phi}^{(\mathrm{P})}$ and
$\big(\boldsymbol{\Phi}^{(\mathrm{P})}\big)^{\mathrm{H}}$.
The required eigenvalue upper bounds can also be evaluated
without explicitly forming the above $N^2\times N^2$ matrices
by exploiting their structures. Simple row/column-sum bounds
can be used where applicable~\cite{2007gershgorin}. These
PAPR-related operations remain within
$\mathcal{O}(L_{\mathrm{P}}N^2)$ per iteration. The projection/normalization and NSP operations in the initialization stage are lightweight. The subsequent DCS and RCS updates require $\mathcal{O}(|\mathcal{D}|L_{\phi}+|\mathcal{R}|)$ operations, whereas the RCS-only designs require only $\mathcal{O}(|\mathcal{R}|)$ operations. These costs are negligible relative to the surrogate construction for a fixed alphabet size. Hence, the dominant per-iteration complexity remains $\mathcal{O}\left(\left(|\mathcal{A}|\log(N) +L_{\mathrm{P}}N\right)N\right)$. \textcolor{black}{For the main simulation setting in Section~\ref{sec:Sim} with $N=128$, $|\mathcal{A}|=152$, and $L_{\mathrm{P}}=4$, the above per-iteration complexity expression gives a numerical scale of approximately $2.02\times10^{5}$ for one basic JIPD-MM update.}

\vspace{-0.3cm}

\section{Numerical Simulations}
\label{sec:Sim}

In this section, we present numerical results to evaluate the proposed AFDM waveform design and JIPD-MM algorithm. We first evaluate the waveform-level ISL and PAPR results under different design modes, then examine the averaged iteration behavior under representative configurations, and finally investigate the resulting sensing and communication performance. Unless otherwise stated, the AFDM-ISAC system and algorithm parameters used throughout this section are summarized in Table~\ref{tab:sim_parameters}. The proposed JIPD-MM algorithm is initialized from a
conventional AFDM waveform. The main optimization iterations terminate when either the
maximum iteration index $r_{\max}$ is reached or
$\frac{
\left\|\mathbf{u}^{(r+1)}-\mathbf{u}^{(r)}\right\|
}{
\left\|\mathbf{u}^{(r)}\right\|
}
\le 10^{-4}$.

\textcolor{black}{Following the definition of $R_{\mathrm{eff}}$ in \eqref{eq:eta_eff}, the comparisons account for both the reduction in DCSs caused by the RCSs and the side-information overhead caused by optimized pre-chirp parameters. For the GPS baseline~\cite{AFDM_PAPR}, whose pre-chirp parameter is allowed to vary across all chirp-subcarriers, the overhead is $B_{\mathrm{SI}}=N\log_2L_{\phi}$. When selected waveform configurations are compared, their $R_{\mathrm{eff}}$ values are chosen to be close whenever applicable. When $|\mathcal{R}|/N$ is swept, the corresponding results instead illustrate the sensing--communication trade-off as the number of RCSs varies. Unless otherwise stated, the simulations assign the RCSs to the highest-index chirp-subcarriers, i.e., $\mathcal{R}=\{N-|\mathcal{R}|,\ldots,N-1\}$.}

\textcolor{black}{We use PSK in the main simulations to maintain the constant-modulus data symbols across the different design modes and avoid introducing additional amplitude variations from the data constellation into the comparisons. However, the proposed formulation is not restricted to PSK constellations. Additional results with nonconstant-modulus QAM data symbols are provided in Appendices~\ref{app:extended_delay} and~\ref{app:qam_extension}.}

\begin{table}[t]
  \centering
  \caption{Main simulation parameters}
  \label{tab:sim_parameters}
  \begin{tabular}{l l}
    \hline
    \textbf{Parameter} & \textbf{Value} \\
    \hline
    Carrier frequency $f_c$ & $28~\text{GHz}$ \\
    System bandwidth $B$ & $100~\text{MHz}$ \\
    Number of AFDM subcarriers $N$ & $128$ \\
    \textcolor{black}{CPP/CP length} & \textcolor{black}{$N/4$}\\
    DAFT parameter $c_1$ & $21/(2N)$ \\
    Pre-chirp alphabet size & $8$ \\
    Side information $b_{\mathrm{SI}}$ & 3\\
    LAZ Delay interval $[-\tau_{\max},\tau_{\max}]$ & $[-8, 8]$\\
    LAZ Doppler interval $[\mu_{\min},\mu_{\max}]$ & $[-4, 4]$ \\
    Number of Doppler grid points $L_{\mu}$ & 9\\
    $w_{\tau,\mu_{q}}$ for $ (\tau,\mu_q)\in\mathcal{A}\backslash\{(0,0)\}$ &
    1 \\
    PAPR oversampling factor $L_{\mathrm{P}}$ & $4$ \\[0.2em]
    PAPR $\ell$-norm exponent $\ell$ & $16$ \\[0.2em]
    Maximum initialization-stage iterations $r_{\mathrm{I}}^{\max}$ & $30$ \\[0.2em]
    Maximum JIPD-MM iterations $r_{\max}$ & $\textcolor{black}{300}$ \\[0.2em]
    \hline
  \end{tabular}
\end{table}

The AFDM waveforms compared in this section are listed below, where the bold names (and their abbreviations) are also used as legend labels:
\begin{itemize}
  \item \textbf{Conventional AFDM}: conventional AFDM waveform with random data symbols~\cite{AFDMOR}, used both as a benchmark and as the initialization of the proposed JIPD-MM algorithm.

  \item \textbf{GPS PAPR minimization~\cite{AFDM_PAPR}}: AFDM waveform with reduced PAPR via GPS method proposed in~\cite{AFDM_PAPR}.

  \item \textbf{Proposed AF shaping (RCS)} and
        \textbf{Proposed PAPR minimization (RCS)}: Only the RCS symbols $\{u[m]\}_{m\in\mathcal{R}}$ are optimized, while the pre-chirp parameters are fixed. In each case, a single metric (ISL or PAPR) is optimized.

  \item \textbf{Proposed AF shaping (RCS+$c_{2,m}$)} and
        \textbf{Proposed PAPR minimization (RCS+$c_{2,m}$)}: both RCS symbols and pre-chirp parameters are optimized for a single metric (ISL or PAPR).

  \item \textbf{Proposed AF shaping and PAPR control (RCS)} and \textbf{Proposed AF shaping and PAPR control (RCS+$c_{2,m}$)}: the RCS symbols are optimized either alone or jointly with the per-subcarrier pre-chirp parameters for the joint AF shaping and PAPR control mode under a prescribed PAPR constraint.
\end{itemize}

\textcolor{black}{Additional comparisons with OCDM and OFDM are provided in Appendix~\ref{app:waveform_extension}.}

\captionsetup[figure]{labelfont={color=black},textfont={color=black}}
\begin{figure*}[t!]
  \color{black}
  \centering
  \includegraphics[width=1.73\columnwidth,
 trim=0 83 0 08,clip]{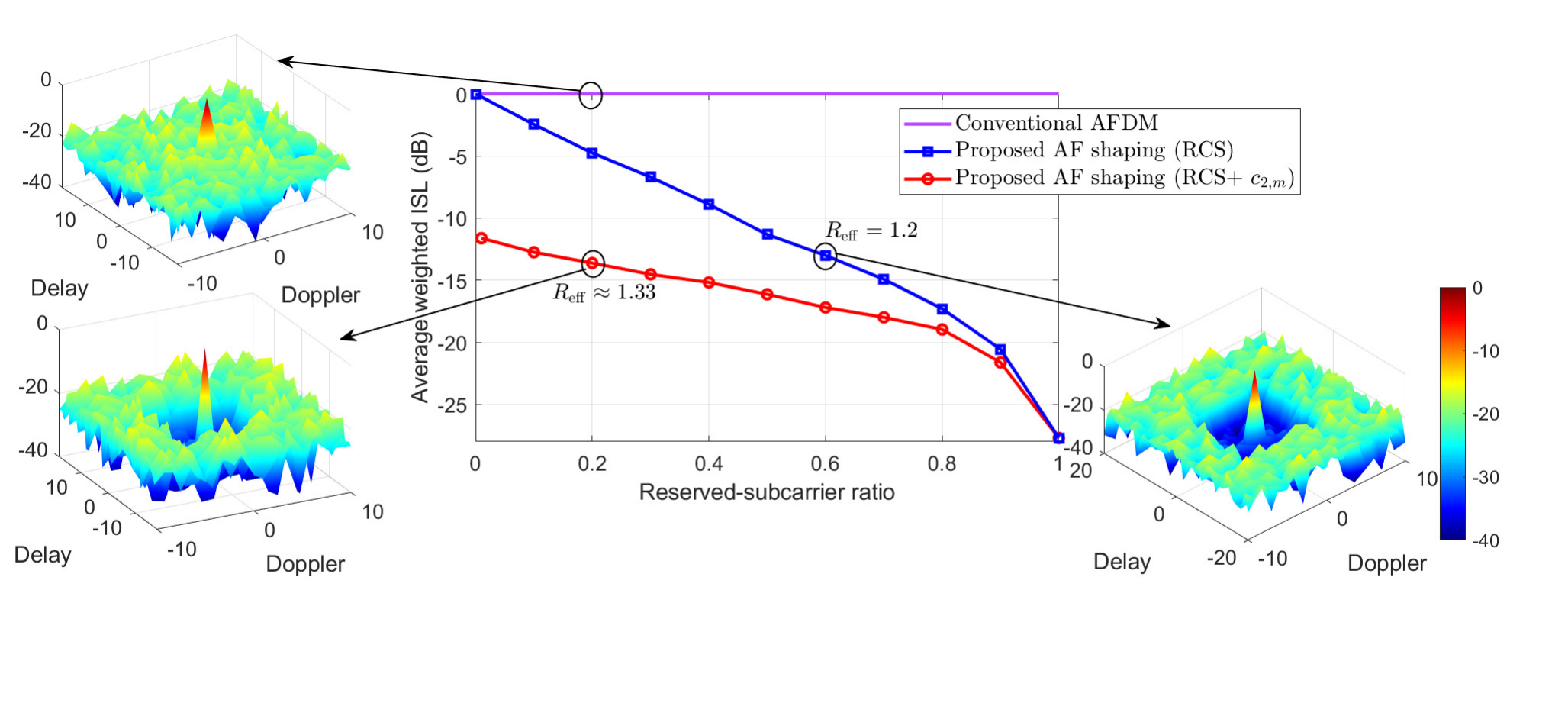}%
  \caption{Average weighted ISL versus reserved-chirp-subcarrier ratio
  $|\mathcal{R}|/N$ for the proposed AF shaping designs, together with representative AF examples.}
  \vspace{-0.3cm}
  \label{fig:ISL_vs_R}
  \color{black}
\end{figure*}

\subsection{AF and PAPR}
\label{subsec:waveform_level}

We first examine the waveform-level characteristics of the proposed AFDM designs in terms of weighted ISL and PAPR.

\subsubsection{AF Shaping Mode}
Fig.~\ref{fig:ISL_vs_R} evaluates the proposed AF shaping mode in terms of the average weighted ISL within the LAZ, where each point is averaged over 100 independent trials. For each sampled value on the $|\mathcal{R}|/N$ axis, the nearest integer value of $|\mathcal{R}|\geq1$ is used. All waveforms in this figure use 8PSK data symbols, and the average weighted ISL is reported in dB relative to conventional AFDM with 8PSK. Conventional AFDM with QPSK was also tested for the fair effective-spectral-efficiency comparisons below but is omitted from this figure, as its average weighted ISL differs from that of conventional AFDM with 8PSK by less than $0.003$~dB. For both proposed designs, increasing $|\mathcal{R}|/N$ reduces the average weighted ISL, illustrating the trade-off between AF shaping and effective spectral efficiency.

\textcolor{black}{The AF examples in Fig.~\ref{fig:ISL_vs_R} compare the two proposed designs with similar average ISL reductions. The RCS+$c_{2,m}$ design is selected at $|\mathcal{R}|/N\approx0.2$, with $R_{\mathrm{eff}}\approx1.33$. Its weighted ISL is reduced by $12.94$~dB in the displayed AF example, while the average reduction on the curve is $13.62$~dB. The RCS-only design is selected at $|\mathcal{R}|/N\approx0.6$, with $R_{\mathrm{eff}}\approx1.20$. Its weighted ISL is reduced by $12.51$~dB in the displayed AF example, while the average reduction on the curve is $13.01$~dB. The two operating points achieve similar average ISL reductions, while the RCS+$c_{2,m}$ design maintains a slightly higher effective spectral efficiency. This indicates that the RCS-only design is a simpler option when avoiding side information overhead is important, whereas the RCS+$c_{2,m}$ design achieves a similar ISL reduction with a slightly higher effective spectral efficiency. An additional AF shaping example with $N=1024$ and larger delay interval is provided in Appendix~\ref{app:extended_delay}.}

\subsubsection{PAPR Minimization Mode}
\begin{figure}[tb]
  \centering
  \subfigure[$R_{\mathrm{eff}}\approx2$ configuration.]{
  \includegraphics[width=0.93\columnwidth,
 trim=32 0 28 8,clip]{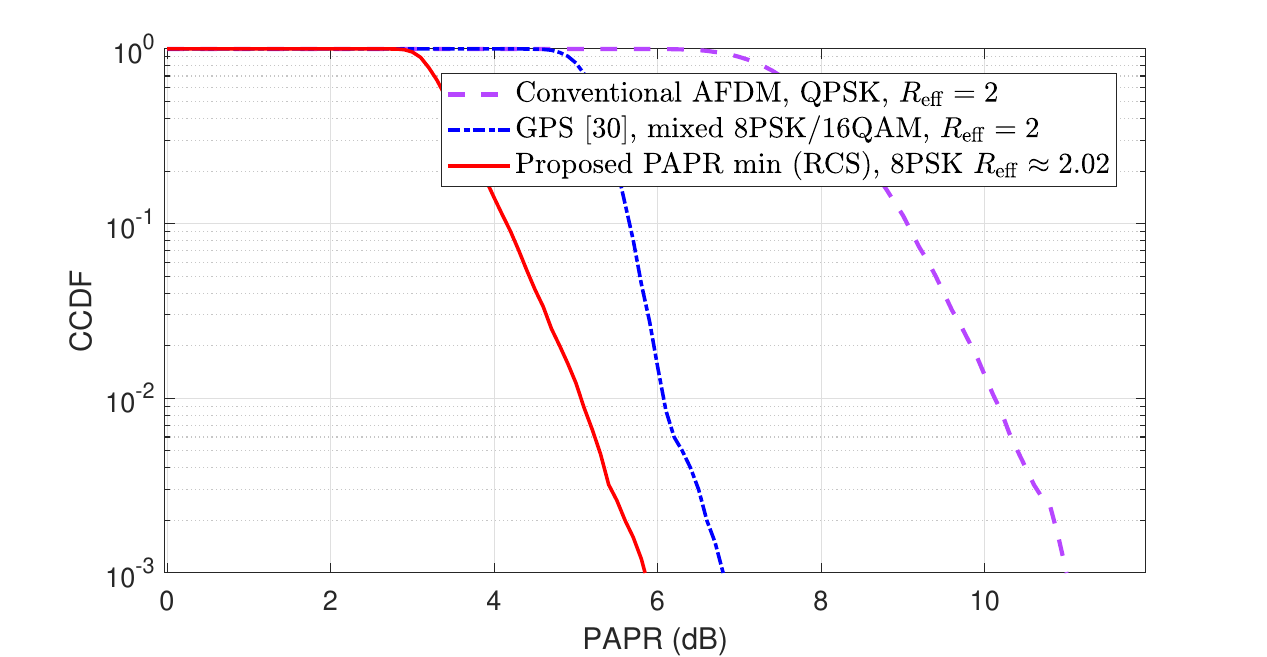}
  \label{fig:PAPR_CCDF_Reff2}}

  \subfigure[$R_{\mathrm{eff}}=1$ configuration.]{
  \includegraphics[width=0.93\columnwidth,
 trim=32 0 28 8,clip]{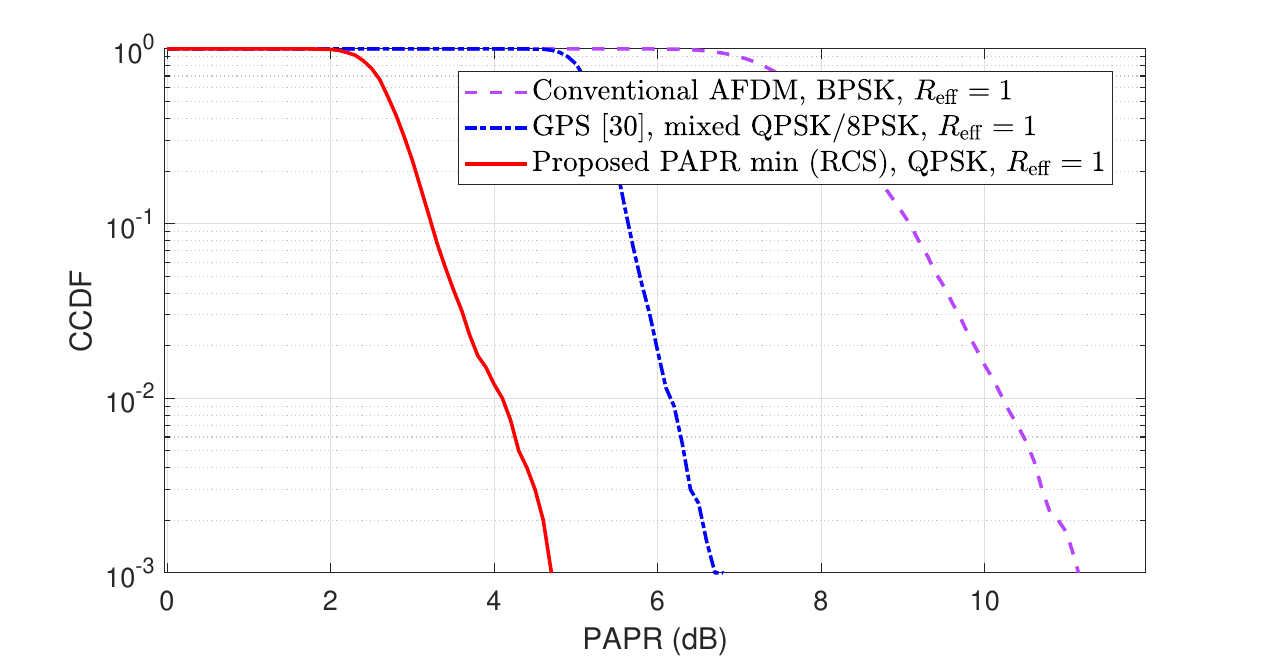}
  \label{fig:PAPR_CCDF_Reff1}}
  \caption{PAPR CCDF comparison for the PAPR minimization mode under two effective-spectral-efficiency configurations.}
  \vspace{-0.3cm}
  \label{fig:PAPR_CCDF}
\end{figure}

\textcolor{black}{Fig.~\ref{fig:PAPR_CCDF} compares the PAPR CCDF of conventional AFDM, the GPS baseline~\cite{AFDM_PAPR}, and the proposed RCS-only design for PAPR minimization. The GPS baseline is implemented with one full sweep over all chirp-subcarriers, where the pre-chirp parameter is updated once for each chirp-subcarrier within an alphabet having the same size as that of the proposed designs. Two effective-spectral-efficiency configurations are considered. In Fig.~\ref{fig:PAPR_CCDF_Reff2}, conventional AFDM uses QPSK and no RCSs, giving $R_{\mathrm{eff}}=2.00$. GPS uses mixed 8PSK/16QAM data symbols, with 64 chirp-subcarriers using 8PSK and the other 64 using 16QAM; its side-information symbols use 16QAM, giving $R_{\mathrm{eff}}=2.00$. The proposed RCS-only design uses 8PSK with $|\mathcal{R}|/N\approx0.33$, giving $R_{\mathrm{eff}}\approx2.02$. In Fig.~\ref{fig:PAPR_CCDF_Reff1}, conventional AFDM uses BPSK, giving $R_{\mathrm{eff}}=1.00$. GPS uses mixed QPSK/8PSK data symbols, with 64 subcarriers using QPSK and the other 64 subcarriers using 8PSK; its side-information symbols use QPSK, giving $R_{\mathrm{eff}}=1.00$. The proposed RCS-only design uses QPSK with $|\mathcal{R}|/N=0.5$, also giving $R_{\mathrm{eff}}=1.00$. Under both configurations, the proposed RCS-only design provides a clear PAPR reduction compared with conventional AFDM and GPS at the same effective spectral efficiency. As shown in the additional comparison in Appendix~\ref{app:papr_comparison}, the design that uses both reserved subcarriers and pre-chirp parameters does not achieve a lower PAPR than the RCS-only design at similar effective spectral efficiency. Therefore, the RCS-only design is adopted for the PAPR-minimization results presented.}

\subsubsection{Joint AF Shaping and PAPR Control Mode}
\begin{figure*}[tb]
  \centering
  \includegraphics[width=1.58\columnwidth,
  trim=27 0 28 0,clip]{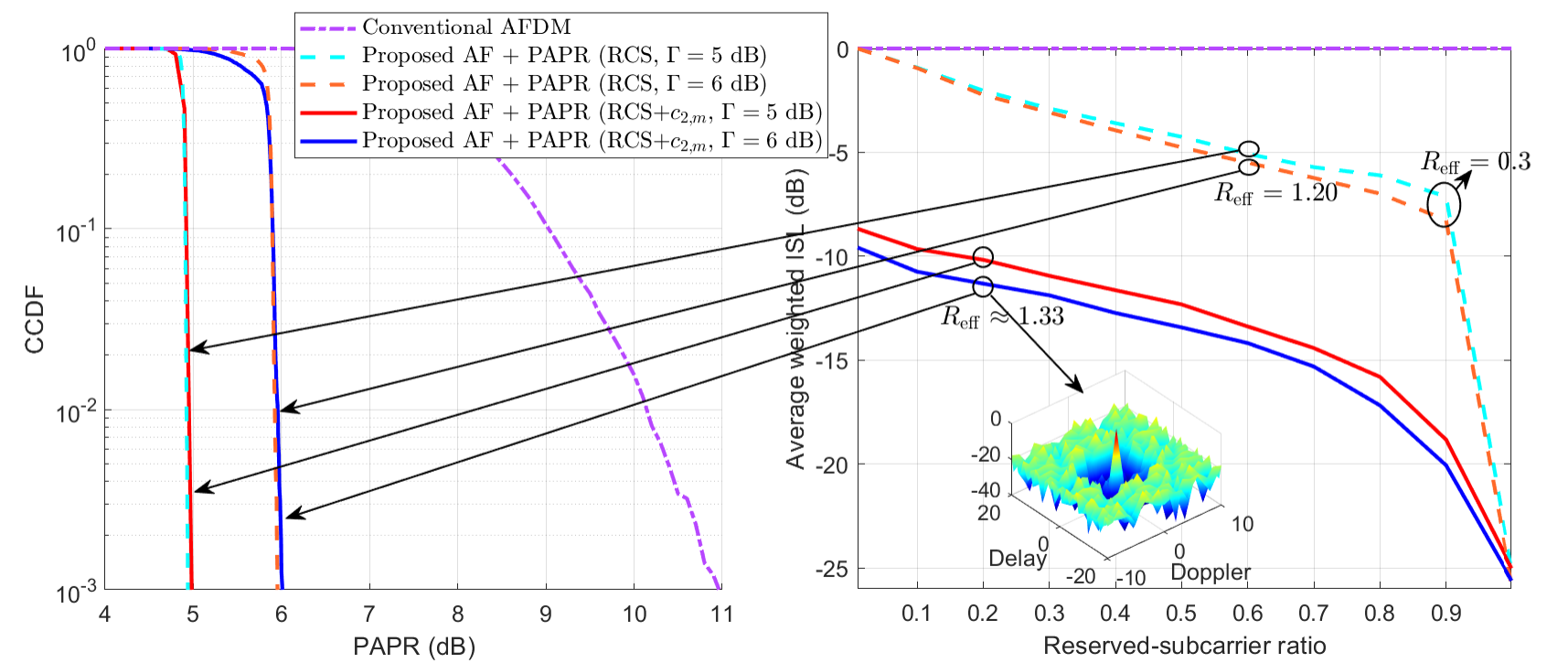}
  \caption{Average weighted ISL versus reserved-chirp-subcarrier ratio
  $|\mathcal{R}|/N$ for the joint AF shaping and PAPR control mode, together with representative PAPR CCDF examples. The PAPR constraints are set as $\Gamma=5$~dB and $\Gamma=6$~dB.}
  \vspace{-0.4cm}
  \label{fig:PAPR_CCDFISLRr}
\end{figure*}

\textcolor{black}{Fig.~\ref{fig:PAPR_CCDFISLRr} presents the waveform characteristics of the joint AF shaping and PAPR control mode for $\Gamma=5$~dB and $\Gamma=6$~dB. For each constraint, the RCS-only and RCS+$c_{2,m}$ designs are evaluated. The average weighted ISL is reported in dB relative to conventional AFDM. The representative PAPR CCDF examples use $|\mathcal{R}|/N\approx0.2$ and $R_{\mathrm{eff}}\approx1.33$ for the RCS+$c_{2,m}$ design, and $|\mathcal{R}|/N\approx0.6$ and $R_{\mathrm{eff}}\approx1.20$ for the RCS-only design. The PAPR CCDF examples verify their PAPR-control capability at these representative settings. The average-weighted-ISL curves further show that the RCS+$c_{2,m}$ design achieves a larger ISL reduction than the RCS-only design under comparable effective spectral efficiency. This trend is consistent with the AF shaping result in Fig.~\ref{fig:ISL_vs_R}. The additional per-subcarrier pre-chirp design makes the advantage more pronounced when AF shaping and PAPR control are considered jointly. Accordingly, the representative comparisons of this mode below use the RCS+$c_{2,m}$ design.}

\subsubsection{\textcolor{black}{Iteration Behavior}}
\label{subsec:iteration_behavior}

\textcolor{black}{Fig.~\ref{fig:JIPD_iteration} examines the iteration behavior of the proposed designs under representative configurations with similar effective spectral efficiencies. The AF shaping and PAPR minimization modes employ the RCS-only design with $|\mathcal{R}|/N\approx0.55$ and $R_{\mathrm{eff}}\approx1.34$, whereas the joint AF shaping and PAPR control mode employs the RCS+$c_{2,m}$ design with $|\mathcal{R}|/N\approx0.2$, $R_{\mathrm{eff}}\approx1.33$, and $\Gamma=5$~dB. All configurations use 8PSK data symbols, and the curves are averaged over 300 independent data realizations.}

\begin{figure}[tb]
\centering

\subfigure[Average weighted ISL.]{
\includegraphics[width=0.88\columnwidth]{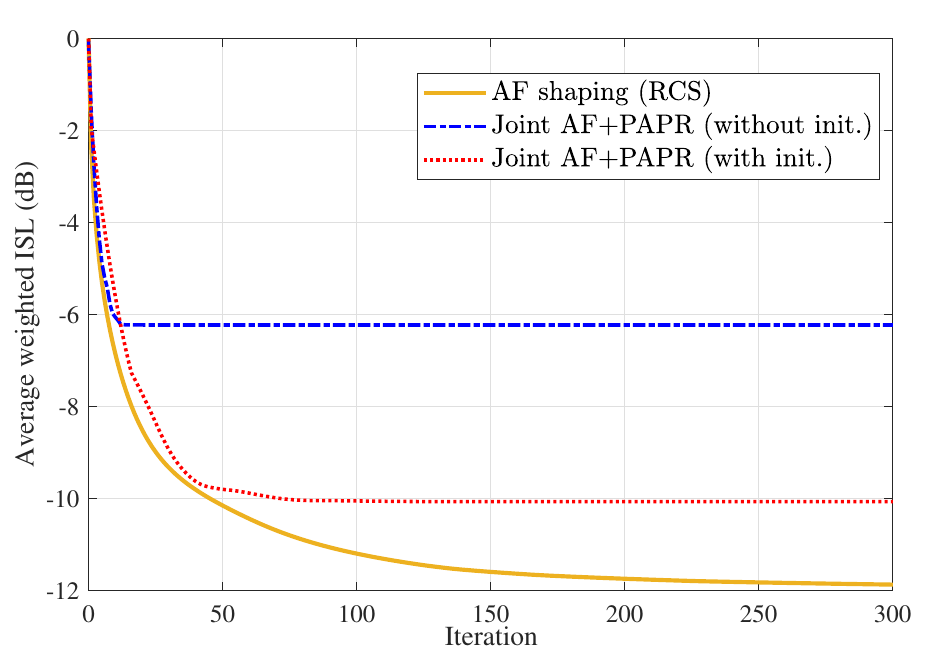}
\label{fig:JIPD_iteration_ISL}}

\subfigure[PAPR.]{
\includegraphics[width=0.88\columnwidth]{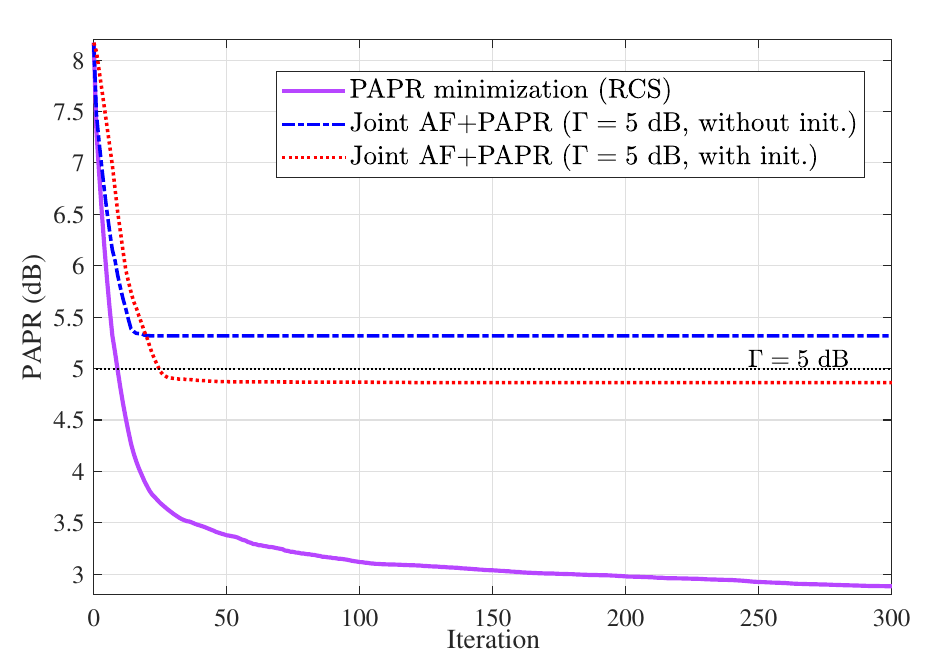}
\label{fig:JIPD_iteration_PAPR}}

\caption{Average iteration behavior of the proposed waveform designs.}
\vspace{-0.3cm}
\label{fig:JIPD_iteration}
\end{figure}

\textcolor{black}{For the two single-objective modes, most of the performance improvement is obtained within the first approximately 50 iterations. Although the strict iterate-based stopping criterion is not always met before $r_{\max}$, the corresponding performance metrics become nearly stable considerably earlier. The AF-shaping and PAPR-minimization modes require on average approximately 253.5 and 183.7 total iterations, respectively, with approximately 63\% and 21\% of the realizations reaching $r_{\max}=300$.}

\textcolor{black}{For the joint AF shaping and PAPR control mode, monotonic decrease of the individual ISL and PAPR metrics is not expected because the two criteria are optimized jointly. Nevertheless, the curves averaged over the independent data realizations show a tendency toward stable levels. With the initialization stage, most of the improvement in the averaged ISL and PAPR occurs within approximately 50--70 total iterations. Moreover, compared with the case without initialization stage, the design with initialization stage achieves both a lower weighted ISL and a lower PAPR, with the average PAPR satisfying the prescribed $\Gamma=5$~dB requirement, demonstrating the importance of the initialization stage. Under the same stopping criterion, the design with initialization stage requires on average approximately 204.4 total iterations, with approximately 35\% of the realizations reaching $r_{\max}=300$, despite the much earlier stabilization observed in the averaged curves.}

\subsection{Sensing and Communication Performance}
\label{subsec:sensing_perf}

We then evaluate whether the above improvements translate into sensing and communication performance gains. To avoid the impact of error propagation across different scenarios, for the communication evaluations involving optimized pre-chirp parameters, the corresponding side information is assumed to be correctly recovered, while its transmission overhead is included in $R_{\mathrm{eff}}$.

For all communication evaluations below, the discrete-time transmit samples are passed through a memoryless Rapp PA model~\cite{Rapp}. Let $s_{\mathrm{PA,in}}[n]$ and $s_{\mathrm{PA,out}}[n]$ denote the PA input and output samples, respectively. The baseband-equivalent static nonlinearity is
\begin{equation}
 s_{\mathrm{PA,out}}[n]
 =
 \frac{s_{\mathrm{PA,in}}[n]}
 {\left(
  1 + \left(\frac{A_{\mathrm{in}}}{A_{\mathrm{sat}}}\right)^{2p}
 \right)^{\tfrac{1}{2p}}},
 \qquad A_{\mathrm{in}}=|s_{\mathrm{PA,in}}[n]|,
 \label{eq:rapp_model}
\end{equation}
where $A_{\mathrm{sat}}$ denotes the saturation amplitude and $p$ controls the smoothness of the transition to saturation. Each waveform realization is normalized to unit average power before amplification, and $A_{\mathrm{sat}}=1$ and $p=3$ are used for all compared waveforms.

\subsubsection{AF Shaping Mode}

\textcolor{black}{We first evaluate weak-target detection performance. A two-target scenario is considered, where the strong target has an SNR $10$~dB higher than that of the weak target. Detection is performed using a CA-CFAR detector~\cite{4102622}, whose threshold is estimated from the eight surrounding delay--Doppler bins. Conventional AFDM with QPSK data symbols is used as the baseline, with $R_{\mathrm{eff}}=2.00$. The proposed RCS-only AF shaping design uses 8PSK data symbols with $|\mathcal{R}|/N\approx0.33$ and $R_{\mathrm{eff}}\approx2.02$.}

\begin{figure}[tb]
  \color{black}
  \centering
  \includegraphics[width=0.88\columnwidth,
trim=12 0 10 10,clip]{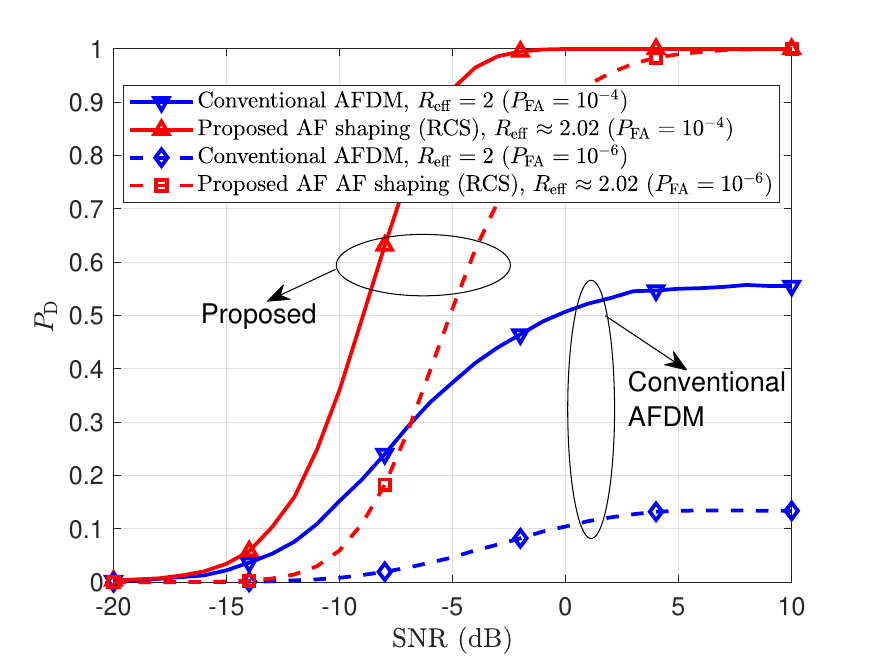}
  \caption{Weak-target detection probability versus SNR for the AF shaping mode.}
  \vspace{-0.3cm}
  \label{fig:weakPd_vs_SNR}
  \color{black}
\end{figure}

\begin{figure}[tb]
  \color{black}
  \centering
  \includegraphics[width=0.88\columnwidth,
trim=12 0 10 10,clip]{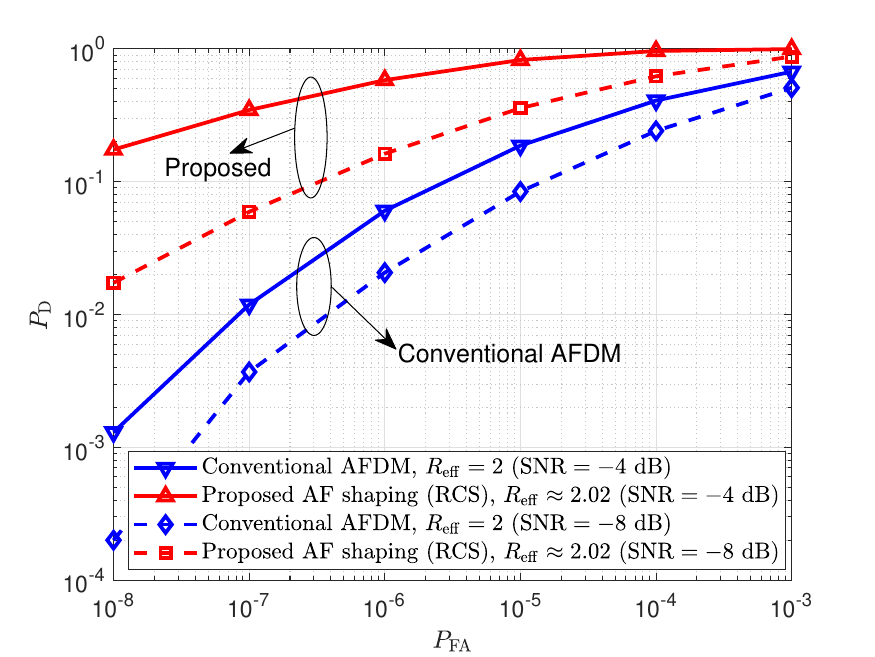}
  \caption{ROC curves of weak-target detection for the AF shaping mode.}
  \vspace{-0.3cm}
  \label{fig:weakROC}
  \color{black}
\end{figure}

\textcolor{black}{Figs.~\ref{fig:weakPd_vs_SNR} and~\ref{fig:weakROC} show the resulting weak-target detection performance. Fig.~\ref{fig:weakPd_vs_SNR} reports $P_{\mathrm{D}}$ as a function of sensing SNR under $P_{\mathrm{FA}}=10^{-4}$ and $P_{\mathrm{FA}}=10^{-6}$. For both false-alarm probabilities, the proposed AF shaping design achieves a higher $P_{\mathrm{D}}$ than conventional AFDM. Fig.~\ref{fig:weakROC} further compares the ROC curves at sensing SNRs of $-8$~dB and $-4$~dB. At both SNR values, the proposed design provides a higher $P_{\mathrm{D}}$ for the same $P_{\mathrm{FA}}$. These results show that the AF sidelobe reduction can effectively improve weak-target detectability.}

\subsubsection{PAPR Minimization Mode}
\label{subsec:comm_perf}

\textcolor{black}{We next evaluate communication performance under nonlinear PA. The same two waveform configurations as those in Fig.~\ref{fig:PAPR_CCDF} are used: the $R_{\mathrm{eff}}\approx2$ configuration in Fig.~\ref{fig:PAPR_CCDF_Reff2} is evaluated over an AWGN channel, while the $R_{\mathrm{eff}}=1$ configuration in Fig.~\ref{fig:PAPR_CCDF_Reff1} is evaluated over a doubly selective channel. For the doubly selective fading case, a three-path time-varying multipath channel is considered. The path delays are fixed at $[2,3,4]$ samples, and the average path-power profile is $[0,-5,-10]$~dB. For each channel realization, the complex path gains are independently generated according to the corresponding average path powers. The normalized Doppler shift of the $p$th path is generated as $\nu_p=\nu_{\max}\cos\theta_p$, where $\theta_p\sim\mathcal{U}[-\pi,\pi]$ and $\nu_{\max}=3$. Perfect channel state information is assumed at the receiver, and MMSE detection is applied.}

\begin{figure}[tb]
  \centering
  \subfigure[AWGN channel, using the $R_{\mathrm{eff}}\approx2$ configuration in Fig.~\ref{fig:PAPR_CCDF_Reff2}.]{
  \includegraphics[width=0.85\columnwidth]{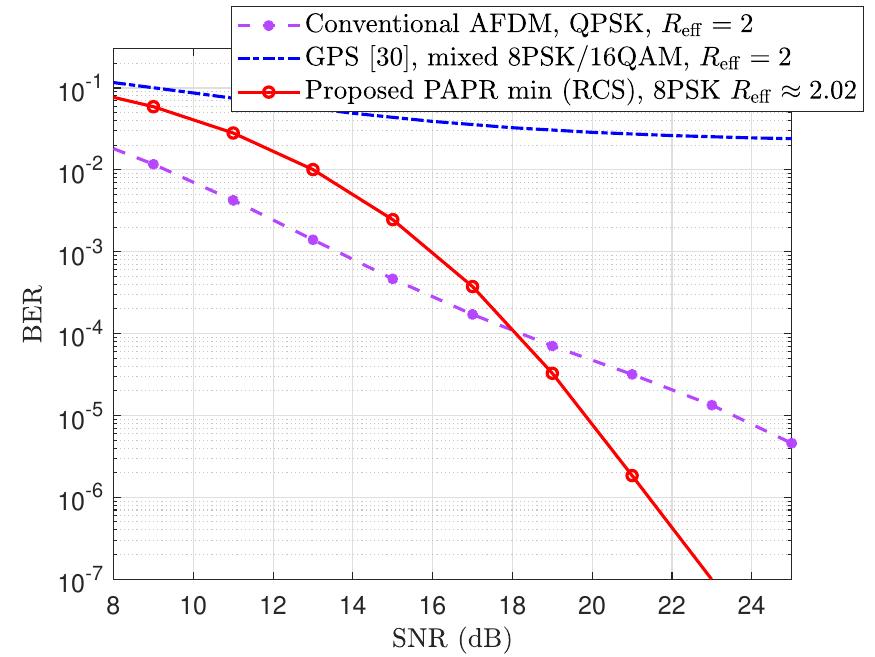}
  \label{fig:PAPR_BER_AWGN}}

  \subfigure[Doubly selective channel, using the $R_{\mathrm{eff}}=1$ configuration in Fig.~\ref{fig:PAPR_CCDF_Reff1}.]{
  \includegraphics[width=0.85\columnwidth]{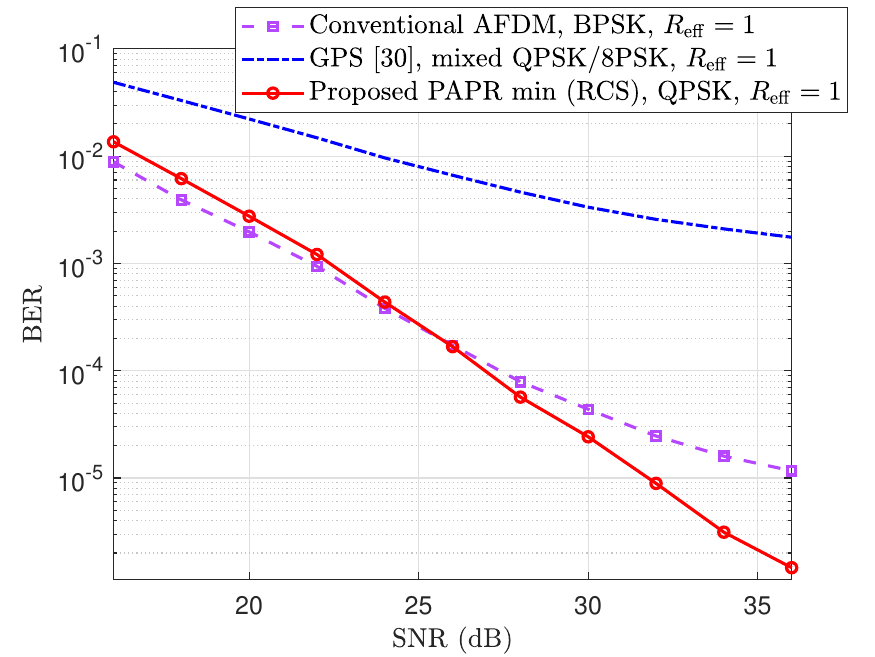}
  \label{fig:PAPR_BER_DS}}
  \caption{BER performance with nonlinear PA for the PAPR minimization mode.}
  \vspace{-0.4cm}
  \label{fig:PAPR_BER}
\end{figure}

\textcolor{black}{Figs.~\ref{fig:PAPR_BER_AWGN} and~\ref{fig:PAPR_BER_DS} show the BER performance under the $R_{\mathrm{eff}}\approx2$ configuration with AWGN channel and the $R_{\mathrm{eff}}=1$ configuration with doubly selective channel, respectively. In both cases, the proposed RCS-only design is close to or slightly worse than conventional AFDM in the lower-SNR region, where conventional AFDM benefits from a lower-order constellation. As SNR increases, nonlinear PA distortion becomes more influential, and the proposed design achieves a lower BER than conventional AFDM despite using a higher-order constellation. It also achieves a substantially lower BER than GPS over the tested SNR range at the same effective spectral efficiency. These results indicate that the PAPR reduction translates into a BER advantage when nonlinear PA distortion becomes dominant.}

\subsubsection{\textcolor{black}{Joint AF Shaping and PAPR Control Mode}}

\textcolor{black}{The preceding results have shown that the proposed framework can effectively suppress AF sidelobes and reduce PAPR, as demonstrated in Figs.~\ref{fig:ISL_vs_R} and~\ref{fig:PAPR_CCDF}, and that these waveform-level improvements can translate into sensing and communication performance gains, as shown in Figs.~\ref{fig:weakPd_vs_SNR}--\ref{fig:weakROC} and Fig.~\ref{fig:PAPR_BER}. The remaining question is how the joint AF shaping and PAPR control mode performs compared with the two single-objective modes of the proposed framework. Since conventional AFDM and GPS do not provide a direct AFDM-ISAC benchmark that jointly addresses AF sidelobe shaping and PAPR control, the AF shaping mode and the PAPR minimization mode are used here as internal single-objective references. All three waveforms use 8PSK data symbols and have similar effective spectral efficiencies (same as in
Fig.~\ref{fig:JIPD_iteration}). Weak-target detection is evaluated at $P_{\mathrm{FA}}=10^{-6}$, while BER is evaluated over an AWGN channel with the same nonlinear PA model.}

\begin{figure}[tb]
\centering
\subfigure[Weak-target detection probability versus SNR.]{
\includegraphics[width=0.85\columnwidth,
trim=10 0 10 13,clip]{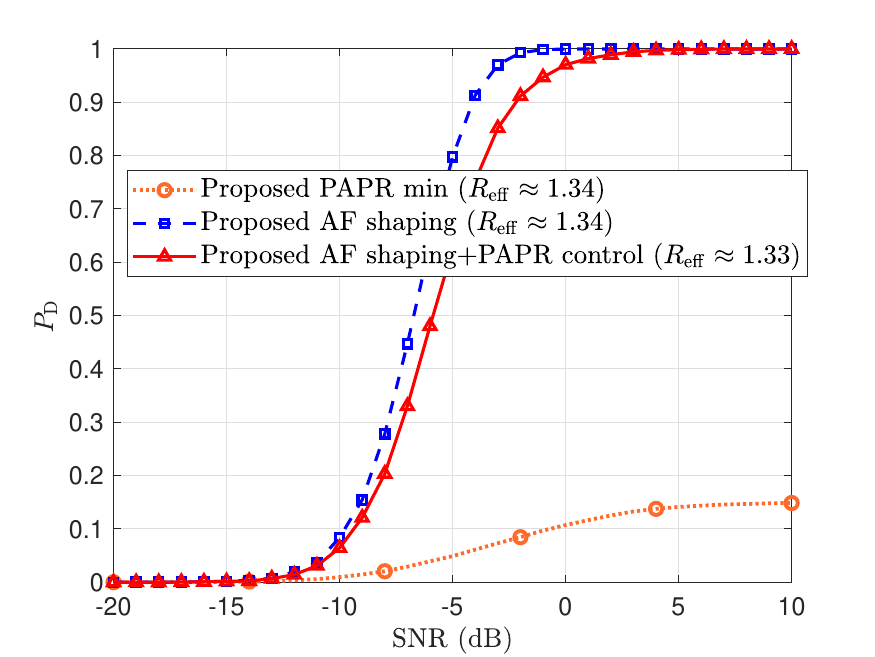}
\label{fig:joint_pd_snr}}
\hfill
\subfigure[BER under nonlinear PA distortion.]{
\includegraphics[width=0.85\columnwidth,
trim=10 0 10 13,clip]{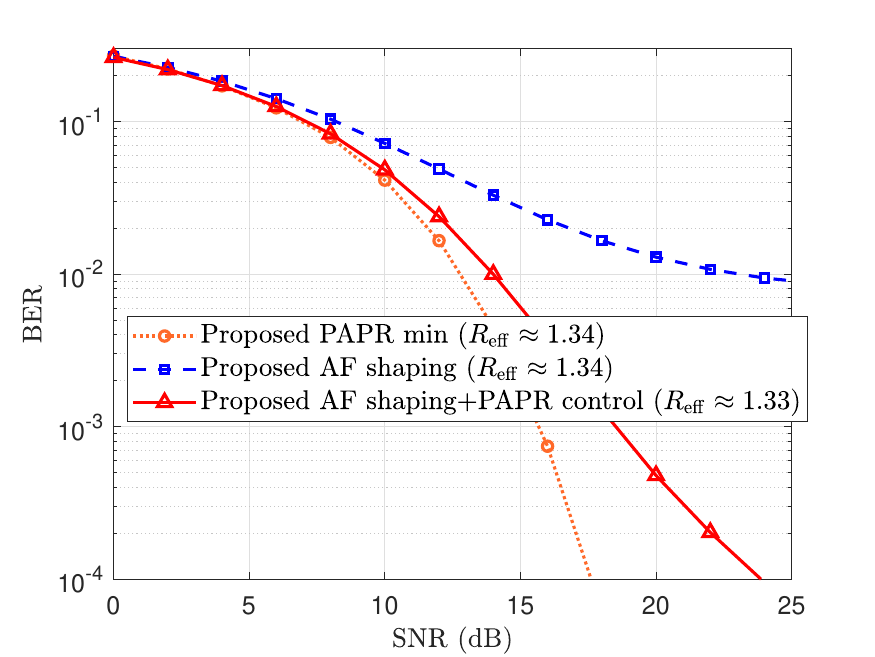}
\label{fig:joint_ber}}
\caption{Detection and communication performance comparison of the proposed joint AF shaping and PAPR control design with the two proposed single-objective designs under similar $R_{\mathrm{eff}}$.}
\vspace{-0.4cm}
\label{fig:joint_performance}
\end{figure}

\textcolor{black}{Figs.~\ref{fig:joint_pd_snr} and~\ref{fig:joint_ber} show the sensing--communication trade-off among the three modes. AF shaping provides the highest weak-target detection probability, whereas PAPR minimization gives the best BER. The joint mode retains most of the detection gain while improving the BER relative to AF shaping, providing a balanced trade-off under comparable effective spectral efficiencies.}

\vspace{-0.2cm}

\section{Conclusion}

In this work, we investigated flexible AFDM-ISAC waveform design for AF shaping and PAPR control. We developed a unified optimization framework that flexibly exploits the RCS symbols and per-subcarrier pre-chirp parameters through the same formulation. Three operating modes for AF shaping, PAPR minimization, and joint AF shaping and PAPR control were considered to support different design priorities. To handle the nonconvex problem with discrete-phase constraints, we iteratively replace it with tractable MM surrogate problems. Numerical results demonstrate clear AF-sidelobe and PAPR reductions across the different operating modes, with the benchmark comparisons performed at similar effective spectral efficiencies. The two single-objective modes respectively favor weak-target detection and BER performance under nonlinear PA distortion, while the joint mode achieves ISL reduction under the prescribed PAPR constraint and enables a more flexible balance between sensing and communication performance.

Future research directions include developing lower-complexity implementations and statistical-metric-based design approaches.
\vspace{-0.2cm}
\appendices
\section{\textcolor{black}{AF Shaping over an Extended Delay Range}}
\label{app:extended_delay}

\textcolor{black}{We further consider an AF shaping example with $N=1024$, 16QAM data symbols, and the zero-Doppler delay interval $\tau\in[-200,200]$. The CPP/CP length is set to $N/4=256$ samples, while the remaining parameters follow Table~\ref{tab:sim_parameters}. The RCS design uses $|\mathcal{R}|/N\approx0.55$, while the RCS+$c_{2,m}$ design uses $|\mathcal{R}|/N\approx0.18$; both configurations have $R_{\mathrm{eff}}\approx1.80$. The results are averaged over 100 independent realizations. At $B=100$~MHz, the considered delay interval corresponds to a relative-range interval of $[-300,300]$~m. As shown in Fig.~\ref{fig:extended_delay}, both proposed designs substantially suppress the zero-Doppler AF sidelobes throughout the considered delay interval.}

\begin{figure}[tb]
\color{black}
\centering
\includegraphics[width=0.8\columnwidth,
 trim=8 2 5 13,clip]{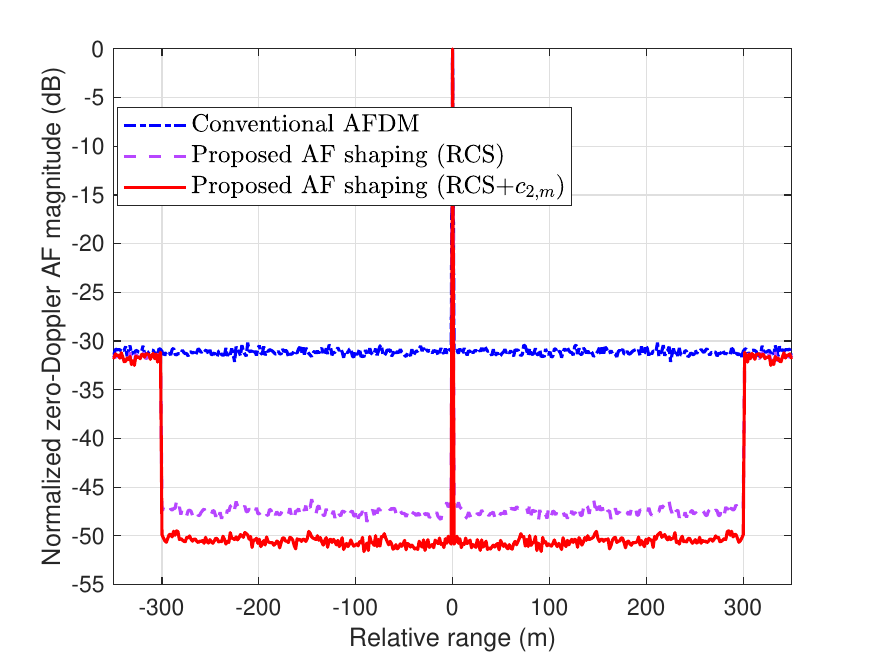}
\caption{Average normalized zero-Doppler AF magnitude with $N=1024$ and 16QAM data symbols.}
\vspace{-0.3cm}
\label{fig:extended_delay}
\color{black}
\end{figure}

\begin{figure}[t]
\color{black}
\centering
\includegraphics[width=0.9\columnwidth,
 trim=25 0 28 8,clip]
{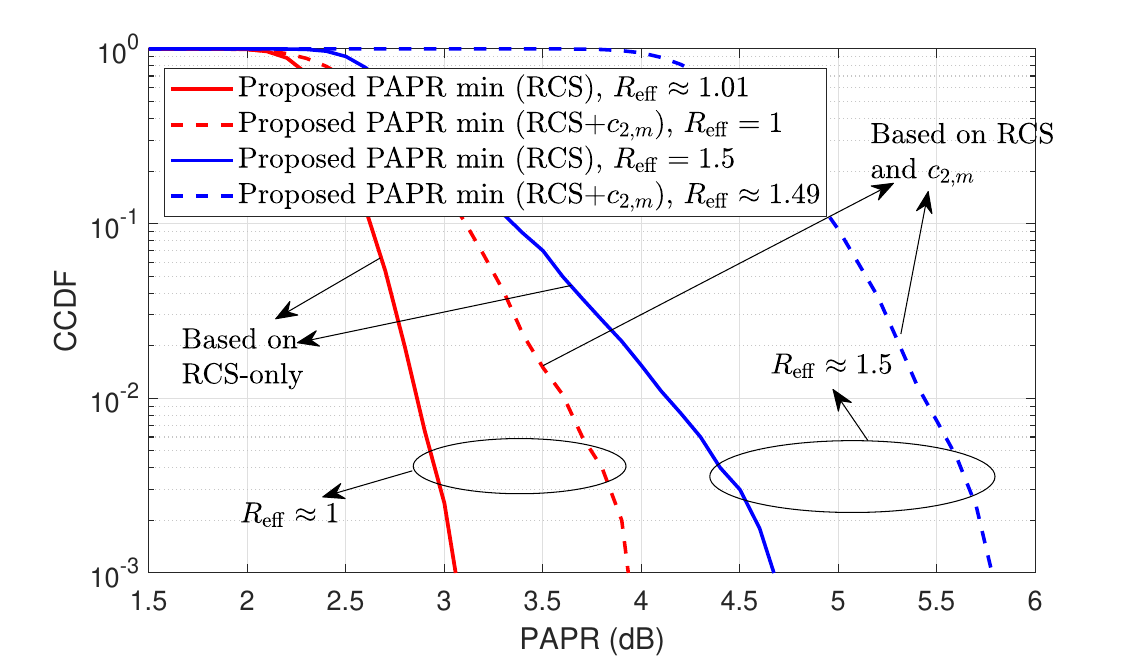}
\caption{Additional PAPR CCDF comparison between the RCS-only
and RCS+$c_{2,m}$ designs.}
\vspace{-0.3cm}
\label{fig:app_papr_comparison}
\color{black}
\end{figure}

\vspace{-0.2cm}
\section{\textcolor{black}{Additional Comparison for PAPR Minimization}}
\label{app:papr_comparison}

\textcolor{black}{Fig.~\ref{fig:app_papr_comparison} compares the RCS-only and RCS+$c_{2,m}$ designs using 8PSK data symbols. The
two pairs are evaluated at $R_{\mathrm{eff}}\approx1$ and $R_{\mathrm{eff}}\approx1.5$, respectively. Under both configurations, the RCS-only design achieves a lower PAPR CCDF than the RCS+$c_{2,m}$ design with similar $R_{\mathrm{eff}}$. Therefore, when only PAPR minimization is considered, using the available RCS symbols for waveform design is more effective than optimizing the per-subcarrier pre-chirp parameters.}
\vspace{-0.2cm}
\section{\textcolor{black}{Extension to OFDM and OCDM}}
\label{app:waveform_extension}

\begin{figure}[t]
\centering

\subfigure[Weak-target detection probability.]{
\includegraphics[width=0.78\columnwidth,
 trim=10 2 10 18,clip]{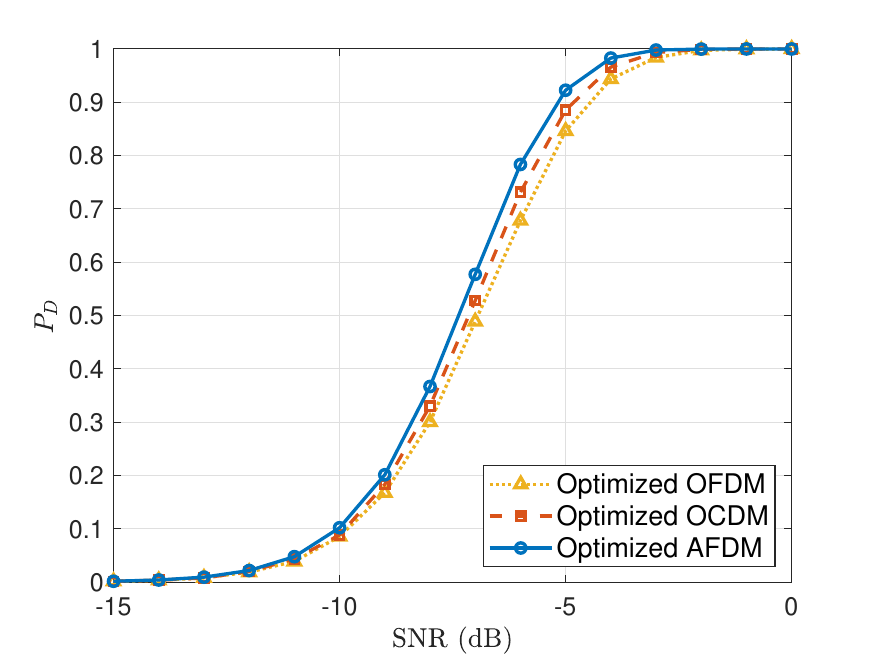}
\label{fig:waveform_extension_pd}}
\vspace{-0.2cm}
\subfigure[BER over the doubly selective channel.]{
\includegraphics[width=0.78\columnwidth,
 trim=10 2 10 18,clip]{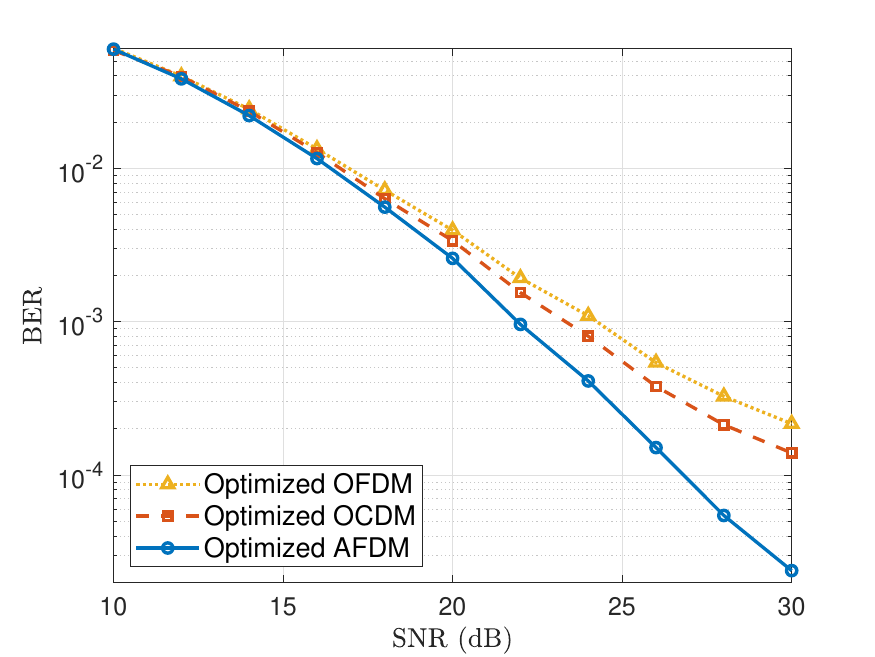}
\label{fig:waveform_extension_ber}}

\caption{Comparison of the optimized AFDM, OCDM, and OFDM waveforms at the same effective spectral efficiency.}
\vspace{-0.2cm}
\label{fig:waveform_extension}
\end{figure}

\textcolor{black}{The proposed RCS-based framework can also be applied to OCDM and OFDM by fixing the chirp parameters according to the corresponding modulation matrices. We consider OCDM with $c_1=c_2=1/(2N)$, and OFDM with $c_1=c_2=0$. All three waveforms use $N=128$, QPSK data symbols, $|\mathcal{D}|=|\mathcal{R}|=64$, and hence $R_{\mathrm{eff}}=1$. The DCSs and RCSs are uniformly interleaved across the full transform-domain band. For sensing, the RCS-only AF shaping design is evaluated using the same two-target detection setting as in Section~\ref{subsec:sensing_perf}, with $P_{\mathrm{FA}}=10^{-6}$. For communication, the RCS-only PAPR minimization designs are evaluated over the same doubly selective channel and PA setting as in Section~\ref{subsec:comm_perf}. As shown in Fig.~\ref{fig:waveform_extension}, under the considered configuration, the optimized AFDM provides a slightly higher weak-target detection probability and a lower BER than the optimized OCDM and OFDM. The BER advantage is consistent with the enhanced robustness of AFDM to doubly selective channels~\cite{AFDMOR}.}

\begin{table}[ht]
\color{black}
\centering
\captionsetup{font={small,color=black}}
\caption{AFDM waveform-design results with 16QAM data symbols. The reported metrics are the average weighted ISL for AF shaping mode and the PAPR for PAPR minimization mode.}

\label{tab:qam_extension}
\begin{tabular}{lcc}
\hline
Design mode & Initial (dB) & Optimized (dB) \\
\hline
AF shaping        & $0$     & $-11.22$ \\
PAPR minimization & $8.16$  & $3.08$ \\
\hline
\end{tabular}
\color{black}
\end{table}

\section{\textcolor{black}{Extension to QAM Data Modulation}}
\label{app:qam_extension}

\textcolor{black}{To verify the applicability of the proposed waveform design framework to nonconstant-modulus data symbols, we further consider the RCS-only AF shaping and PAPR minimization designs with 16QAM. We set $|\mathcal{R}|/N=0.5$, yielding $R_{\mathrm{eff}}=2$. The remaining waveform parameters follow Table~\ref{tab:sim_parameters}. Table~\ref{tab:qam_extension} reports the results averaged over 100 independent data realizations, where the weighted ISL is normalized with respect to the corresponding initial waveform. The results show weighted-ISL and PAPR reductions comparable to those obtained with PSK. Another example in Appendix~\ref{app:extended_delay} also employs 16QAM data symbols.}

\bibliographystyle{IEEEtran}
\bibliography{BibOne.bib}

\end{document}